\documentclass[11pt]{article}
\usepackage[utf8]{inputenc}
\usepackage[english]{babel}
\usepackage{amsmath}
\usepackage{amscd,amsfonts,amstext,
	amsthm,amssymb,color,graphicx,commath,
	latexsym,mathrsfs,tocloft,tikz,tkz-fct,listofsymbols,dsfont,subcaption,comment,bm, pdfcomment}
\usepackage{booktabs}
\usepackage[round]{natbib}
\usepackage{titlesec}
\usepackage[letterpaper,left=2.5cm,top=1.54cm,bottom=2.54cm,right=2.5cm]{geometry}

\usepackage[left, pagewise]{lineno}

\usepackage{epstopdf}
\usepackage{color}
\usepackage{multirow}
\usepackage{enumerate}
\usepackage{authblk}
\usepackage{hyperref}
\usepackage{pdfpages}
\usepackage{lineno}
\usepackage{mathtools}

\titleformat*{\section}{\large\bfseries}

\newcommand{\1}{\mathds{1}}

\newcommand{\E}{\mathbb{E}}
\newcommand{\Prob}{\mathbb{P}}

\newcommand{\R}{\mathbb{R}}
\newcommand{\V}{\mathbb{V}}
\newcommand{\Cov}{\mathbf{C}ov}

\newcommand{\rmnum}[1]{\romannumeral #1}
\newcommand{\Rmnum}[1]{\uppercase\expandafter{\romannumeral #1\relax}}
\DeclareMathOperator*{\argmin}{arg\,min}
\DeclareMathOperator*{\arginf}{arg\,inf}

\DeclareMathOperator*{\esssup}{ess\,sup}
\DeclareMathOperator*{\essinf}{ess\,inf}

\newtheorem{corollary}{Corollary}[section]

\newtheorem{example}{Example}[section]
\newtheorem{lemma}{Lemma}[section]
\newtheorem{proposition}{Proposition}[section]
\newtheorem{remark}{Remark}[section]
\newtheorem{theorem}{Theorem}[section]

\numberwithin{equation}{section}

\newtheorem{assumption}{Assumption}

\usepackage{hyperref}

\definecolor{darkblue}{rgb}{0.1,0.1,0.7}
\definecolor{darkred}{rgb}{0.9,0.1,0.1}
\definecolor{brown}{rgb}{0.55,0.27,0.07}

\hypersetup{colorlinks, allcolors= darkblue}

\def\brown{\color{brown}}

\title{\vspace{+0.5cm}\Large\textbf{{Extremal Mean-Variance Functionals over Wasserstein Balls: Applications to Risk Sharing}}}

\author[1]{Wenjun Jiang}
\author[2]{Yiying Zhang}
\author[3]{Zhenfeng Zou}
\affil[1]{\emph{\small Department of Mathematics and Statistics, University of Calgary, Calgary, AB, T2N 1N4, Canada. E-mail: \url{wenjun.jiang@ucalgary.ca}}}
\affil[2]{\emph{\small Department of Mathematics, Southern University of Science and Technology, Shenzhen 518055, P.R. China. E-mail: \url{zhangyy3@sustech.edu.cn}}}
\affil[3]{\emph{\small School of Public Policy and Management, University of Science and Technology of China, Hefei 230026, P.R. China. E-mail: \url{zfzou@ustc.edu.cn}}}
\date{\vspace{-1cm}}

\begin{document}
	\maketitle
	
\begin{abstract}
We characterize the worst- and best-case values of a mean-variance functional over a 2-Wasserstein ball. Using quantile representations and the geometry of attainable means and standard deviations, we reduce both infinite-dimensional problems to scalar equations and construct the extremal laws as location-scale transformations of the reference distribution. Their values depend on the reference law only through its first two moments. We then derive dual representations and formulate proportional risk sharing under heterogeneous beliefs as a finite-dimensional optimization problem. Under homogeneous beliefs, we show that the classical proportional allocation remains optimal for every ambiguity radius and $\alpha$-maxmin weight. Finally, we study a coupled distortion-variance functional and characterize its worst-case quantile through a convex-envelope construction, allowing the extremal law to change in shape as well as location and scale.
\end{abstract}




\section{Introduction}\label{intro}
The mean-variance criterion, introduced by \citet{Markowitz1952}, provides a simple and tractable framework for balancing the expected level of an outcome against its variability. Because it depends only on the first two moments, it is also closely connected to expected utility under normally distributed returns or quadratic preferences \citep{Levy1979}. This tractability, however, comes with a well-known sensitivity to model inputs: the mean and variance must be specified or estimated, and small errors in these quantities can lead to substantial changes in the resulting decisions. As \citet{Michaud1989} observed, mean-variance optimization may therefore act as an ``estimation-error maximizer''. This concern has motivated a broad literature on robust mean-variance optimization \citep{lobo2000worst, pflug2007ambiguity}. In this paper, we take a distributional perspective and study the evaluation of a fixed exposure under model uncertainty. Taking a reference distribution as a benchmark, we ask how large or small its mean-variance value can become when the underlying distribution varies within a 2-Wasserstein neighborhood. We formulate this question as a pair of extremal distributional problems over a 2-Wasserstein ball and subsequently apply the resulting characterizations to risk sharing.

Distributional uncertainty is commonly modeled either through moment information \citep[see, e.g.,][]{chen2011tight, Cai2023Distributionally, zuo2025worst} or through neighborhoods around a reference distribution. Among the latter approaches, the 2-Wasserstein distance provides a natural way to control distributional deviations while allowing both the mean and dispersion to vary from their benchmark values. \citet{mohajerin2018data} established a general duality framework for Wasserstein distributionally robust optimization and derived finite-dimensional convex reformulations for a broad class of problems. Subsequent studies have applied Wasserstein ambiguity to, among others, robust optimized certainty equivalents \citep{bartl2020computational} and worst-case moment estimation under partial distributional information \citep{tang2023worst}. More closely related to our setting, \citet{blanchet2022distributionally} study mean-variance portfolio selection by minimizing worst-case variance subject to a robust expected-return constraint, and derive regularized reformulations together with data-driven procedures for selecting the ambiguity radius and return target. Our analysis takes a different perspective: we place the mean and variance directly in a single objective and evaluate them jointly under the same candidate distribution. We then characterize both the upper and lower extrema of this mean-variance functional over a 2-Wasserstein ball, together with the corresponding extremal distributions.

The mean-variance functional leads to a nonstandard Wasserstein optimization problem because variance depends nonlinearly on the underlying distribution. Moreover, the Wasserstein constraint couples the feasible perturbations of the mean and standard deviation, so separate sharp bounds for these moments do not determine the extrema of their weighted combination. We therefore reduce the extremal problem to an optimization over the jointly attainable mean-standard deviation pairs induced by the Wasserstein ball. This geometric reduction shows that the worst- and best-case laws are location-scale transformations of the reference distribution. Their objective values depend on the reference law only through its first two moments, whereas the extremal laws inherit its distributional shape. We also obtain scalar dual representations that make the resulting extremal evaluations directly applicable to risk-sharing problems. The same geometric characterization yields comparative statics with respect to variance aversion: increasing variance aversion raises the variance and lowers the mean of the worst-case law, with the reverse pattern for the best-case law.

Our extremal distributional analysis also provides a natural tool for studying risk sharing under distributional ambiguity. In mean-variance risk sharing, agents allocate an aggregate loss while balancing the expected losses and variances of their individual positions. The literature has established several structural results for such problems. \citet{barrieu2005inf} derived proportional allocation rules for dilated risk measures, while \citet{acciaio2007optimal} studied optimal risk sharing for non-monotone monetary functionals, including the mean-variance criterion. More recently, \citet{Bernard2025ambiguity} showed that, under a planner's worst-case criterion, distributional ambiguity preserves the optimal proportional allocation in the mean--variance setting, and \citet{Blier-Wong_Lauzier_2026} investigated comonotonic improvements under feasibility constraints. These studies primarily characterize how aggregate risk should be allocated. A complementary issue concerns the valuation of a given allocation under model uncertainty: even when the same allocation remains optimal across candidate models, its associated cost can vary substantially with the underlying distribution. Our extremal results address this issue by identifying the distributions within a Wasserstein ambiguity set that generate the largest and smallest costs for a given allocation.

Both extrema become relevant when the planner adopts an $\alpha$-maxmin criterion, which combines worst- and best-case evaluations to represent different attitudes toward ambiguity. Best-case evaluations also arise naturally in distributionally favorable optimization \citep{Jiang2024Distributionally}. Related $\alpha$-maxmin formulations have been studied in mean--variance portfolio selection \citep{yu2020portfolio} and in dynamically consistent decision models \citep{Beissner2020}. We use our extremal characterizations to analyze two risk-sharing settings that differ in both the admissible allocations and the way distributional evaluations are aggregated. Under heterogeneous beliefs, the planner restricts attention to proportional allocations and minimizes the sum of the agents' worst-case evaluations, allowing each agent to use her own reference distribution and Wasserstein radius. The dual representations reduce this problem to a finite-dimensional optimization over the sharing proportions and dual variables, with each reference distribution entering only through its mean and variance. Under homogeneous beliefs, the planner instead considers all feasible allocations under an $\alpha$-maxmin criterion. We show that the classical proportional allocation is optimal under every candidate distribution and therefore remains optimal for the $\alpha$-maxmin problem. Hence, neither the Wasserstein radius nor the ambiguity parameter affects the optimal sharing rule, although both influence the associated extremal distributions and aggregate costs.

The mean-variance criterion captures distributional uncertainty only through the first two moments. To move beyond this moment-based evaluation, we further consider the broader class of distortion riskmetrics introduced by \citet{wang2020distortion}, which contains distortion risk measures as special cases. Extremal problems for distortion-based functionals have received increasing attention under various forms of partial distributional information. \citet{shao2023distortion, shao2023Extreme} derived sharp bounds for distortion risk measures under moment constraints, with and without distributional symmetry. Incorporating Wasserstein ambiguity, \citet{bernard2024robust} used isotonic projection to characterize worst- and best-case distortion risk measures when the first two moments are known, both with and without an additional 2-Wasserstein constraint. For the broader class of distortion riskmetrics, \citet{Pesenti2025} established conditions linking the robustification of a non-convex distortion riskmetric to that of an associated convex counterpart. More recently, \citet{liu2025robust} extended the Wasserstein-based extremal analysis to distortion riskmetrics and incorporated additional structural information such as unimodality and symmetry.

Against this background, we study the extremal behavior of a coupled distortion-variance functional, in which a distortion riskmetric and a variance penalty are evaluated under the same candidate distribution. This coupling fundamentally changes the structure of the extremal problem. Unlike the mean, which depends only on the first moment, a distortion riskmetric is sensitive to the entire distribution and can distinguish between laws with identical means and variances. Consequently, the location--scale structure obtained in the mean-variance case generally breaks down, and changes in distributional shape must also be taken into account. Under suitable regularity conditions, we characterize the worst-case quantile through a convex-envelope construction. The function being convexified jointly incorporates the distortion function, the variance-aversion parameter, and the reference quantile, thereby capturing their interaction in determining the extremal law. The resulting worst-case distribution may therefore differ from the reference law not only in location and scale but also in shape. We illustrate this mechanism with Range Value-at-Risk (RVaR), showing how the Wasserstein radius and variance aversion jointly reshape the worst-case quantile.

The rest of the paper is organized as follows. Section~\ref{model} introduces the necessary preliminaries and formulates the extremal problems. Section~\ref{sec:robust} characterizes the extremal distributions and their corresponding values. Section~\ref{sec:alpha-maxmin-riskshare} applies these results to two risk-sharing problems. Section~\ref{sec:extension} develops the extension to distortion riskmetrics. Section~\ref{sec:conclusion} concludes. Appendix~\ref{sec:A} collects auxiliary results used in the proofs, while Appendix~\ref{sec:high_dim} extends the one-dimensional analysis of Section~\ref{sec:robust} to the multivariate setting.


\section{Problem setup}\label{model}


Let $(\Omega,\mathscr{F},\Prob)$ be an atomless probability space and $\mathcal{L}^2$ the class of real-valued random variables on this space with finite second moments. Write $\mathcal{F}^2$ for their cumulative distribution functions (CDFs) under $\Prob$. For a random variable $X\in\mathcal{L}^2$ with CDF $F \in \mathcal{F}^2$ (written as $X \sim F$), its left- and right-continuous quantile functions are defined, respectively, as
$$F^{-1}(t)=\inf\left\{x\in\R: \Prob(X \leq x)\ge t\right\}, ~ t \in (0, 1],$$
and 
$$F^{-1+}(t)=\inf\left\{x\in\R: \Prob(X \leq x) > t\right\}, ~ t \in [0, 1),$$
with the conventions $F^{-1}(0) := F^{-1+}(0) = \essinf(X)$ and $F^{-1+}(1) := F^{-1}(1) = \esssup(X)$.
We write $\mu_F=\E_F[X]$ and $\sigma_F^2=\V_F[X]$ for the expectation and variance of $X$ under $F$. Positive values denote losses and negative values profits.
For two random variables $X$ and $Y$, we write $X \leq_{\rm cx} Y$ (respectively, $X {\leq_{\rm icx}} Y$) and say that $X$ is smaller than $Y$ in the convex (resp. {increasing convex}) order if $\E[\phi(X)] \leq \E[\phi(Y)]$ for every convex (resp. {increasing convex}) function $\phi: \R \to \R$ for which the expectations exist, with expectations taken under the respective distributions.
For an event $A \in \mathscr{F}$, we denote by $\1_A$ its indicator function, i.e., $\1_A(\omega) = 1$ if $\omega \in A$ and 0 otherwise.
Let $\Delta_n$ denote the standard simplex $\Delta_n := \big\{(w_1, \ldots, w_n) \in \R^n: w_i \geq 0,\ \sum_{i=1}^n w_i=1 \big\}$.

\subsection{Wasserstein ambiguity sets}

We describe distributional uncertainty by a neighborhood of a reference distribution.
The Wasserstein distance gives this neighborhood a probabilistic and geometric interpretation \citep{villani2009optimal} and allows the distributions to have different supports \citep{pesenti2020portfolio}. For $p \geq 1$, let $K_1, K_2$ be two one-dimensional CDFs on $\R$ with $\int_\R |x|^p d K_i(x) < \infty$ for $i=1,2$. 
The $p$-Wasserstein distance between them is defined as
\begin{equation}\label{Wass}
    W_p(K_1,K_2) = \left(\inf_{Y_1\sim K_1,Y_2\sim K_2}\ \E[|Y_1-Y_2|^p]\right)^{\frac{1}{p}},
\end{equation}
where the infimum is taken over all joint distributions of $(Y_1, Y_2)$ with marginal distributions $K_1$ and $K_2$, respectively.
It is well known that the infimum in \eqref{Wass} is always attainable \citep{villani2009optimal}.
Equivalently, by the well-known quantile representation of the Wasserstein distance, we have
\begin{equation}\label{Wass:quantile}
    W_p(K_1,K_2)=\left(\int_0^1\left|K_1^{-1}(t)-K_2^{-1}(t)\right|^pdt\right)^{\frac{1}{p}},
\end{equation}
a property that will be instrumental in our subsequent analysis.
When $p=2$, the Wasserstein distance captures changes in both location and dispersion and is intimately connected to optimal transport with quadratic cost. 

We assume that the true distribution of $X$ lies within a 2-Wasserstein ball of radius $\sqrt{\epsilon}$ centered at a known reference distribution $G\in\mathcal F^2$. Throughout, we assume $\epsilon>0$ and positive variance for the underlying one-dimensional reference distributions, i.e., $\sigma_G>0$. The corresponding ambiguity set is defined as
\begin{equation}
    \mathcal{W}^2(G;\epsilon)=\left\{F\ \text{is a CDF}:\ W_2(F,G)\le \sqrt{\epsilon}\right\}.
\end{equation}
\noindent To ease the presentation, we let $\mathcal{M}^2(G;\epsilon)=\mathcal{F}^2\cap\mathcal{W}^2(G;\epsilon)$. Larger values of $\epsilon$ allow greater departures from the reference distribution. Under this ambiguity set, a decision maker (DM) evaluates the random loss $X$ using the mean-variance preference discussed in the next subsection, accounting for all distributions within the Wasserstein ball.

\subsection{Extremal mean-variance preferences}

The DM evaluates the loss $X$ by its expected value plus a variance penalty:
\begin{equation}\label{MV}
    V^F(X) := \E_F[X]+\gamma\V_F[X],
\end{equation}
where $\gamma>0$ measures the DM's variance aversion. This criterion captures the trade-off between the average level of loss and its variability: a higher $\gamma$ reflects greater penalization of uncertainty surrounding the loss. Despite its simplicity, the mean-variance objective remains one of the most widely used objective functions in portfolio selection, insurance pricing, and risk management. Its dependence on only the first two moments makes it analytically tractable. Moreover, the mean-variance objective admits an expected-utility interpretation either when outcomes are normally distributed or when preferences are represented by quadratic utility.

In the presence of distributional ambiguity, the DM does not know which distribution $F \in \mathcal{F}^2$ governs the loss $X$.
A natural robust approach is to evaluate the loss under the worst-case distribution, leading to the problem
\begin{equation}\label{Prob:sup}
    \sup_{F\in\mathcal{M}^2(G;\epsilon)}\ V^F(X),
\end{equation}
and, symmetrically, the best-case scenario
\begin{equation}\label{Prob:inf}
    \inf_{F\in\mathcal{M}^2(G;\epsilon)}\ V^F(X),
\end{equation}
where $V^F(X)$ is given by \eqref{MV}.

The worst-case problem \eqref{Prob:sup} reflects an ambiguity-averse stance: the DM seeks to guard against the most adverse distribution within the Wasserstein ball. The best-case problem \eqref{Prob:inf}, by contrast, represents the optimistic perspective. In the $\alpha$-maxmin criterion of Section~\ref{sec:alpha-maxmin-riskshare}, the DM's ambiguity attitude determines the weights in a convex combination of these worst- and best-case evaluations.

In Section~\ref{sec:robust}, we will study these two problems and give  explicit representations of the extremal quantile functions and the corresponding mean-variance values. As we shall see, the quantile representation of the Wasserstein distance established in \eqref{Wass:quantile} plays a pivotal role in solving these problems.

\section{Extremal mean-variance optimization under Wasserstein ambiguity}\label{sec:robust}

This section solves the extremal optimization problems \eqref{Prob:sup} and \eqref{Prob:inf}. 
We begin by showing that the Wasserstein constraint is always binding at optimality, and then characterize the extremal distributions and the corresponding optimal values for both the worst-case and best-case problems. 
As a natural extension, Appendix~\ref{sec:high_dim} studies the high-dimensional versions of \eqref{Prob:sup} and \eqref{Prob:inf}, which can be reduced to their one-dimensional counterparts by projecting the high-dimensional Wasserstein balls onto one-dimensional Wasserstein balls.

To facilitate the subsequent derivations, recall the following moment identities in terms of $F^{-1}$:
\begin{align*}
    \E_F[X]=\int_0^1F^{-1}(t)dt\quad\text{and}\quad \V_F[X]=\int_0^1\left(F^{-1}(t)\right)^2dt-\left(\int_0^1F^{-1}(t)dt\right)^2.
\end{align*}
These identities will be used repeatedly throughout the analysis. The following lemma establishes that the Wasserstein constraint $W_2(F,G) \le \sqrt{\epsilon}$ is binding for both the worst-case and best-case problems.

\begin{lemma}\label{lem:binding}
    Any optimizer $F^*$ of the worst-case (resp. best-case) problem \eqref{Prob:sup} (resp. \eqref{Prob:inf}), if it exists, must satisfy $W_2(F^*,G) = \sqrt{\epsilon}$.
\end{lemma}
\begin{proof}
    For problem \eqref{Prob:sup}, suppose that $F$ is feasible and its Wasserstein constraint is slack, i.e.,
    $$
    \int_0^1\left|F^{-1}(t)-G^{-1}(t)\right|^2dt<\epsilon.
    $$
    Let $\Tilde{F}$ be a CDF with ${\Tilde{F}}^{-1}(t)=F^{-1}(t)+c$, where $c>0$ is such that
    $$
    \int_0^1\left|F^{-1}(t)+c-G^{-1}(t)\right|^2dt=\epsilon.
    $$
    A direct calculation gives
    \begin{align*}
        &\E_{\Tilde{F}}[X]+\gamma\V_{\Tilde{F}}[X]\\
        &\ ~~ = \int_0^1\left(F^{-1}(t)+c\right)dt+\gamma\left(\int_0^1(F^{-1}(t)+c)^2dt-\left(\int_0^1(F^{-1}(t)+c)dt\right)^2\right) \\
        &\ ~~ = \int_0^1\left(F^{-1}(t)+c\right)dt+\gamma\left(\int_0^1\left(F^{-1}(t)\right)^2dt-\left(\int_0^1F^{-1}(t)dt\right)^2\right) \\
        &\ ~~ > \E_F[X]+\gamma\V_F[X].
    \end{align*}
    The best-case statement follows analogously by shifting in the opposite direction.
\end{proof}

Lemma \ref{lem:binding} allows us to characterize the worst-case distribution $F_{wo}$ in \eqref{Prob:sup}. The quantile representation below constructs the extremal distribution, while the scalar dual representation will be used to solve risk sharing problems.

\begin{theorem}\label{thm:wo_quantile}
Let $\theta^*$ be the unique root of 
$$\tan(\theta)-2\gamma\sigma_G-2\gamma\sqrt{\epsilon}\sin(\theta)=0,~~ \theta \in \left(0, \frac{\pi}{2}\right).$$ 
Then the optimal quantile function of the worst-case distribution $F_{wo}$ for problem \eqref{Prob:sup} is given by
\begin{equation}\label{wo_quantile_form}
    {F_{wo}}^{-1}(t)=(\mu_G+\sqrt{\epsilon}\cos(\theta^*))+(\sigma_G+\sqrt{\epsilon}\sin(\theta^*))\frac{G^{-1}(t)-\mu_G}{\sigma_G}.
\end{equation}
Moreover, the worst-case value of problem \eqref{Prob:sup} is
\begin{eqnarray}
\sup_{F \in {\cal M}^2(G; \epsilon)} V^F(X)=\mu_G + \frac{\lambda^* \gamma}{\lambda^* - \gamma} \sigma_G^2 + \frac{1}{4 \lambda^*}  + \lambda^* \epsilon,  \label{wo_values2}
\end{eqnarray}
where $\lambda^*$ is the unique solution of $\frac{\gamma^2}{(\lambda-\gamma)^2} \sigma_G^2 + \frac{1}{4\lambda^2}=\epsilon$ within $(\gamma,\infty)$.
\end{theorem}
\begin{proof}
The worst-case problem \eqref{Prob:sup} can be rewritten as 
\begin{equation}\label{Prob:main-2nd}
\sup_{\mu_F\in\R}\sup_{\substack{F\in\mathcal{M}^2(G;\epsilon) \\ \E_F[X] = \mu_F}}\ \E_F[X]+\gamma\V_F[X].
\end{equation}
We first focus on the inner problem of \eqref{Prob:main-2nd} over $F$ for a fixed $\mu_F$. 
By Lemma \ref{lem:binding}, the Wasserstein constraint is binding:
\begin{equation*}
\int_0^1\left|F^{-1}(t)-G^{-1}(t)\right|^2dt=\epsilon.
\end{equation*}
Expanding the left-hand side gives
\begin{equation*}
\int_0^1(F^{-1}(t))^2dt-2\int_0^1F^{-1}(t)G^{-1}(t)dt+\int_0^1(G^{-1}(t))^2dt=\epsilon.
\end{equation*}
Consequently, the objective becomes
\begin{align*}
\E_F[X]+\gamma\V_F[X]=&\ \mu_F+\gamma\left(\int_0^1(F^{-1}(t))^2dt-\mu_F^2\right) \\
=&\ \mu_F-\gamma\mu_F^2+\gamma\left(2\int_0^1F^{-1}(t)G^{-1}(t)dt-\int_0^1(G^{-1}(t))^2dt+\epsilon\right).
\end{align*}
Thus, for fixed $\mu_F$, the inner problem of \eqref{Prob:main-2nd} now reduces to
$$
\sup_{F}\ \int_0^1F^{-1}(t)G^{-1}(t)dt,\quad \text{s.t.}\ \int_0^1F^{-1}(t)dt=\mu_F.
$$
Applying the Cauchy-Schwarz inequality yields
\begin{align*}
\int_0^1F^{-1}(t)G^{-1}(t)dt=&\ \int_0^1(F^{-1}(t)-\mu_F)(G^{-1}(t)-\mu_G)dt+\mu_F\mu_G \\
\le&\ \sigma_F\left(\int_0^1(G^{-1}(t)-\mu_G)^2dt\right)^{\frac{1}{2}}+\mu_F\mu_G,
\end{align*}
with equality attained when
\begin{equation}\label{opti_form}
F^{-1}(t)=\mu_F+\sigma_F\frac{G^{-1}(t)-\mu_G}{\sigma_G}.
\end{equation}
With the derived optimal quantile form \eqref{opti_form}, the binding Wasserstein constraint $W_2(F,G) = \sqrt{\epsilon}$ becomes $(\mu_F-\mu_G)^2+(\sigma_F-\sigma_G)^2=\epsilon$. 
Parametrize $\mu_F=\mu_G+\sqrt{\epsilon}\cos(\theta)$ and $\sigma_F=\sigma_G+\sqrt{\epsilon}\sin(\theta)$ with $\theta \in [0, \pi/2]$ (so that both $\sin(\theta)$ and $\cos(\theta)$ are non-negative).
The objective then becomes
$$
\max_{\theta\in[0,\frac{\pi}{2}]}\ L(\theta):=\mu_G+\sqrt{\epsilon}\cos(\theta)+\gamma\left(\sigma_G+\sqrt{\epsilon}\sin(\theta)\right)^2.
$$
We have $L'(\theta)=-\sqrt{\epsilon}\sin(\theta)+2\gamma\sqrt{\epsilon}\sigma_G\cos(\theta)+2\gamma\epsilon\sin(\theta)\cos(\theta)$, with $L'(0+)=2\gamma\sqrt{\epsilon}\sigma_G>0$ and $L'(\pi/2-)=-\sqrt{\epsilon}<0$. Hence every maximizer lies in $(0,\pi/2)$.
The first-order condition $L'(\theta)=0$ is equivalent to
\begin{equation}\label{eq:maximizer_sup_prob}
\tan(\theta)=2\gamma\sigma_G+2\gamma\sqrt{\epsilon}\sin(\theta).
\end{equation}
On $(0,\pi/2)$, $\tan(\theta)$ is strictly convex and $\sin(\theta)$ is strictly concave. The left-hand side of \eqref{eq:maximizer_sup_prob} starts below the right-hand side at zero and diverges as $\theta\uparrow\pi/2$. Thus \eqref{eq:maximizer_sup_prob} has a unique root $\theta^*$, which is the maximizer of $L$.
Thus, \eqref{wo_quantile_form} holds.

It remains to establish the dual representation in \eqref{wo_values2}. 
To this end, note that
\begin{equation}\label{worst-case}
\sup_{F \in {\cal M}^2(G; \epsilon)} V^F(X) = \sup_{F \in {\cal M}^2(G; \epsilon)} \inf_{x \in \R} \E_F\left[X + \gamma (X-x)^2\right].
\end{equation}
Let $g(x, F) := \E_F\left[X + \gamma (X-x)^2\right]$.
Observe that $g(x, F)$ is linear in $F$ and convex in $x$. For every $F\in\mathcal M^2(G;\epsilon)$, $\E_F[X]$ lies in the bounded interval $[x_1,x_2]$, where $x_1=\mu_G-\sqrt\epsilon$ and $x_2=\mu_G+\sqrt\epsilon$; see Remark \ref{remark:intuition1}.
Theorem~\ref{thm:sion} therefore allows us to interchange the supremum and the infimum:
\begin{eqnarray}
\sup_{F \in {\cal M}^2(G; \epsilon)} V^F(X) &=& \sup_{F \in {\cal M}^2(G; \epsilon)} \inf_{x \in [x_1, x_2]} \E_F\left[X + \gamma (X-x)^2\right]   \nonumber  \\
&=& \inf_{x \in [x_1, x_2]} \sup_{F \in {\cal M}^2(G; \epsilon)} \E_F\left[X + \gamma (X-x)^2\right] \nonumber  \\
&=& \inf_{x \in [x_1, x_2]} \inf_{\lambda \geq 0} \E_G\left[\sup_{y \in \R} \left\{y + \gamma (y-x)^2 - \lambda (X-y)^2\right\} + \lambda \epsilon\right]  \nonumber  \\
&=& \inf_{x \in [x_1, x_2]} \inf_{\lambda > \gamma} \E_G\left[\frac{(1 - 2 \gamma x + 2 \lambda X)^2}{4 (\lambda - \gamma)} + \gamma x^2 - \lambda X^2 + \lambda \epsilon\right] \nonumber \\
&=& \inf_{\lambda > \gamma} \inf_{x \in [x_1, x_2]}  \E_G\left[\frac{(1 - 2 \gamma x + 2 \lambda X)^2}{4 (\lambda - \gamma)} + \gamma x^2 - \lambda X^2 + \lambda \epsilon\right], \label{inf_min}
\end{eqnarray}
where the third equality follows from Theorem 2.4 in \cite{bartl2020computational}. For any fixed $\lambda > \gamma$, the objective function of \eqref{inf_min} is strictly convex in $x$ and admits a unique minimizer $x^*(\lambda)=\mu_G + \frac{1}{2\lambda}$. Substituting this minimizer gives
\begin{equation*}
\Phi(\lambda) := \mu_G + \frac{\lambda \gamma}{\lambda - \gamma} \sigma_G^2 + \frac{1}{4 \lambda}  + \lambda \epsilon.
\end{equation*}
$\Phi(\lambda)$ is strictly convex on $(\gamma, \infty)$, with $\lim_{\lambda \to \gamma^+} \Phi(\lambda) = \lim_{\lambda \to \infty} \Phi(\lambda) = \infty$.
The infimum of $\Phi(\lambda)$ is therefore attained at a unique point $\lambda^* \in (\gamma, \infty)$, characterized by
\begin{equation*}
\epsilon = \frac{\gamma^2}{(\lambda^*-\gamma)^2} \sigma_G^2 + \frac{1}{4(\lambda^*)^2}.
\end{equation*}
Hence $1/(2\lambda^*)\leq\sqrt\epsilon$ and $x^*(\lambda^*)\in[x_1,x_2]$, completing the proof.
\end{proof}

The next theorem provides the analogous characterization for the best-case problem \eqref{Prob:inf}.

\begin{theorem}\label{thm:be_quantile}
Let $\theta^*$ be the unique root of 
$$\tan(\theta)-2\gamma\sigma_G-2\gamma\sqrt{\epsilon}\sin(\theta)=0, ~~ \theta \in \left(\pi,\frac{3}{2}\pi\right).$$
Then the optimal quantile function of the best-case distribution $F_{be}$ for problem \eqref{Prob:inf} is given by
\begin{equation}\label{be_quantile_form}
{F_{be}}^{-1}(t)=(\mu_G+\sqrt{\epsilon}\cos(\theta^*))+(\sigma_G+\sqrt{\epsilon}\sin(\theta^*))\frac{G^{-1}(t)-\mu_G}{\sigma_G}.
\end{equation}
Moreover, the best-case value of problem \eqref{Prob:inf} is
\begin{eqnarray}
\inf_{F \in {\cal M}^2(G; \epsilon)} V^F(X)=\mu_G+\frac{\lambda^*\gamma}{\lambda^*+\gamma}\sigma_G^2-\frac{1}{4\lambda^*}-\lambda^*\epsilon,   \label{be_values2}
\end{eqnarray}
where $\lambda^*$ is the unique positive solution of $\frac{\gamma^2\sigma_G^2}{(\lambda+\gamma)^2}+\frac{1}{4\lambda^2}=\epsilon$. 
\end{theorem}
\begin{proof}
Following the same line of reasoning as in the proof of Theorem \ref{thm:wo_quantile}, problem \eqref{Prob:inf} can be rewritten as
\begin{equation}\label{Prob:fixed_mu}
\inf_{\mu_F\in\R}\inf_{\substack{F\in\mathcal{M}^2(G;\epsilon) \\ \E_F[X] = \mu_F}}\ \E_F[X]+\gamma\V_F[X].
\end{equation}
For a fixed $\mu_F$, the inner problem of \eqref{Prob:fixed_mu} takes the form
\begin{equation}\label{Prob:inner2}
\inf_F\ \int_0^1(F^{-1}(t))^2dt,\quad\text{s.t.}\ \int_0^1(F^{-1}(t)-G^{-1}(t))^2dt\le\epsilon,\ \int_0^1F^{-1}(t)dt=\mu_F.
\end{equation}
Feasibility requires $|\mu_F-\mu_G|\leq\sqrt{\epsilon}$. At either endpoint, equality in the Cauchy--Schwarz inequality forces $F^{-1}(t)=G^{-1}(t)+\mu_F-\mu_G$ almost everywhere, so the feasible distribution is unique and already has the required location--scale form.
For $|\mu_F-\mu_G|<\sqrt{\epsilon}$, the translated quantile $F_0^{-1}(t)=G^{-1}(t)+\mu_F-\mu_G$ has mean $\mu_F$ and satisfies
\[
\int_0^1\bigl(F_0^{-1}(t)-G^{-1}(t)\bigr)^2dt
=(\mu_F-\mu_G)^2<\epsilon.
\]
Since \eqref{Prob:inner2} is convex in $F^{-1}$, Slater's condition holds for interior feasible means \citep[see, e.g., Theorem 1 of Section 8.3 of][]{luenberger1997optimization}, and the corresponding inner problem can be analyzed through its Lagrangian:
\begin{equation}\label{Prob:dual}
\inf_F\ \int_0^1(F^{-1}(t))^2dt+\lambda_1\int_0^1(F^{-1}(t)-G^{-1}(t))^2dt-\lambda_2\int_0^1F^{-1}(t)dt
\end{equation}
with $\lambda_1\ge 0$ and $\lambda_2\in\R$. 
Completing the square yields
\begin{align*}
&\int_0^1(F^{-1}(t))^2dt+\lambda_1\int_0^1(F^{-1}(t)-G^{-1}(t))^2dt-\lambda_2\int_0^1F^{-1}(t)dt \\
&\ ~~ = (1+\lambda_1)\int_0^1\left(F^{-1}(t)-\frac{2\lambda_1G^{-1}(t)+\lambda_2}{2(1+\lambda_1)}\right)^2dt-\int_0^1\frac{(2\lambda_1G^{-1}(t)+\lambda_2)^2}{4(1+\lambda_1)}dt+\lambda_1\int_0^1(G^{-1}(t))^2dt.
\end{align*}
Hence, for fixed $\lambda_1\geq0$ and $\lambda_2\in\R$, the solution to \eqref{Prob:dual} is
$${F^*}^{-1}(t)=\frac{\lambda_2}{2(1+\lambda_1)}+\frac{\lambda_1}{1+\lambda_1}G^{-1}(t).$$
Together with the endpoint cases considered above, this shows that the quantile function solving \eqref{Prob:fixed_mu} is an increasing affine transformation of $G^{-1}$:
$$F_{be}^{-1}(t)=\mu_F+\sigma_F\frac{G^{-1}(t)-\mu_G}{\sigma_G}.$$ 
By Lemma \ref{lem:binding}, the Wasserstein constraint $W_2(F,G) \le \sqrt{\epsilon}$ is binding and reduces to $(\mu_F-\mu_G)^2+(\sigma_F-\sigma_G)^2=\epsilon$. 
Similar to the proof of Theorem \ref{thm:wo_quantile}, we write $\mu_F=\mu_G+\sqrt\epsilon\cos(\theta)$ and $\sigma_F=\sigma_G+\sqrt\epsilon\sin(\theta)$ for $\theta\in[\pi,3\pi/2]$, so that both $\sin(\theta)$ and $\cos(\theta)$ are nonpositive.
Problem \eqref{Prob:fixed_mu} then becomes
\begin{equation}
\min_{\theta}\ L(\theta):=\mu_G+\sqrt{\epsilon}\cos(\theta)+\gamma\left(\sigma_G+\sqrt{\epsilon}\sin(\theta)\right)^2.
\end{equation}
Since $L'(\theta)=-\sqrt{\epsilon}\sin(\theta)+2\gamma\sqrt{\epsilon}\sigma_G\cos(\theta)+2\gamma\epsilon\sin(\theta)\cos(\theta)$, with $L'(\pi+)=-2\gamma\sqrt{\epsilon}\sigma_G<0$ and $L'(3\pi/2-)=\sqrt{\epsilon}>0$, we conclude that the minimizer of $L(\theta)$ must be within $(\pi, 3\pi/2)$. 
Note that
\begin{align*}
L'(\theta)=&\ -\sqrt{\epsilon}\sin(\theta)+2\gamma\sqrt{\epsilon}\sigma_G\cos(\theta)+2\gamma\epsilon\sin(\theta)\cos(\theta) \\
=&\ -\sqrt{\epsilon}\cos(\theta)\left(\tan(\theta)-2\gamma\sigma_G-2\gamma\sqrt{\epsilon}\sin(\theta)\right).
\end{align*}
On $(\pi,3\pi/2)$, $\tan(\theta)$ is increasing whereas $\sin(\theta)$ is decreasing. Because the bracketed expression above is negative at $\pi$ and tends to a positive value as $\theta\uparrow3\pi/2$, $L$ has a unique minimizer $\theta^*\in(\pi,3\pi/2)$, characterized by
$$
\tan(\theta)=2\gamma\sigma_G+2\gamma\sqrt{\epsilon}\sin(\theta), ~~ \theta \in \left(\pi,\frac{3}{2}\pi\right).
$$
Note that the above equation guarantees that $\sigma_G+\sqrt{\epsilon}\sin(\theta^*)>0$, which leads to 
$$F_{be}^{-1}(t)=(\mu_G+\sqrt{\epsilon}\cos(\theta^*))+(\sigma_G+\sqrt{\epsilon}\sin(\theta^*))\frac{G^{-1}(t)-\mu_G}{\sigma_G}.$$ Thus, \eqref{be_quantile_form} holds.

It remains to establish \eqref{be_values2}. Let $x_1=\mu_G-\sqrt\epsilon$ and $x_2=\mu_G+\sqrt\epsilon$. Then
\begin{align}
\inf_{F \in {\cal M}^2(G; \epsilon)} V^F(X)
&= \inf_{x \in [x_1,x_2]} \inf_{F \in {\cal M}^2(G; \epsilon)}
\E_F\left[X + \gamma (X-x)^2\right] \nonumber\\
&= \inf_{x \in [x_1,x_2]} \sup_{\lambda \geq 0}
\E_G\left[-\frac{(2\gamma x + 2\lambda X -1)^2}{4(\gamma+\lambda)}
+ \gamma x^2 + \lambda X^2 - \lambda \epsilon\right] \nonumber\\
&= \inf_{x \in [x_1,x_2]} \sup_{\lambda \geq 0}
\Biggl\{\frac{\lambda \gamma}{\lambda+\gamma} x^2
+ \frac{\gamma (1-2 \lambda \mu_G)}{\lambda + \gamma} x \nonumber\\
&\hspace{5em}+ \frac{\lambda \gamma}{\lambda+\gamma} (\mu_G^2+\sigma_G^2)
+ \frac{4 \lambda \mu_G - 1}{4(\lambda+\gamma)} - \lambda \epsilon \Biggr\} \nonumber\\
&= \sup_{\lambda \geq 0} \inf_{x \in [x_1,x_2]}
\Biggl\{\frac{\lambda \gamma}{\lambda+\gamma} x^2
+ \frac{\gamma (1-2 \lambda \mu_G)}{\lambda + \gamma} x \nonumber\\
&\hspace{5em}+ \frac{\lambda \gamma}{\lambda+\gamma} (\mu_G^2+\sigma_G^2)
+ \frac{4 \lambda \mu_G - 1}{4(\lambda+\gamma)} - \lambda \epsilon \Biggr\},
\label{sup_max}
\end{align}
where the last equality holds by Theorem~\ref{thm:sion}, as the objective function is convex in $x$ and concave in $\lambda$, and the optimization over $x$ can be restricted to a compact set without loss of optimality. The minimizer of the inner problem of \eqref{sup_max} is $x^*(\lambda)=\mu_G-\frac{1}{2\lambda}$. This turns the objective of \eqref{sup_max} into
$$
\Phi(\lambda):=\mu_G + \frac{\lambda \gamma}{\lambda+\gamma} \sigma_G^2 - \frac{1}{4\lambda} - \lambda \epsilon.
$$
Since $\lim_{\lambda\to0^+}\Phi(\lambda)=\lim_{\lambda\to\infty}\Phi(\lambda)=-\infty$ and $\Phi(\lambda)$ is concave, the supremum of $\Phi(\lambda)$ is attained at a unique point $\lambda^*\in(0,\infty)$, characterized by
$$
\epsilon=\frac{\gamma^2\sigma_G^2}{(\lambda^*+\gamma)^2}+\frac{1}{4(\lambda^*)^2}.
$$
The proof is now complete.
\end{proof}

\begin{remark}\label{remark:intuition1}
    In Theorems \ref{thm:wo_quantile} and \ref{thm:be_quantile}, both extremal quantiles are increasing affine transformations of the reference quantile. This observation can be understood through the decomposition of the Wasserstein distance established in \cite{bernard2024robust},
    \begin{align*}
        W_2(F,G)^2=&\ \int_0^1(F^{-1}(t)-G^{-1}(t))^2dt \\
        =&\ (\mu_F-\mu_G)^2+(\sigma_F-\sigma_G)^2+2\sigma_F\sigma_G\left(1-\text{corr}(F^{-1}(U),G^{-1}(U))\right),
    \end{align*}
    where $\text{corr}(\cdot,\cdot)$ denotes the Pearson correlation coefficient and $U$ is a uniform random variable on $(0,1)$.
    It follows that
    $$
    \left\{(\mu_F,\sigma_F):\ W_2(F,G)^2\le\epsilon\right\}\subseteq\{(\mu_F,\sigma_F):\ (\mu_F-\mu_G)^2+(\sigma_F-\sigma_G)^2\le\epsilon\}.
    $$
    For a fixed pair $(\mu_F,\sigma_F)$, the Wasserstein distance is minimized when $\text{corr}(F^{-1}(U),G^{-1}(U))=1$. Since the mean–variance objective depends only on $(\mu_F,\sigma_F)$, imposing perfect positive correlation maximizes the feasible region for $(\mu_F,\sigma_F)$. Consequently, it suffices to search for extremal quantile functions among those that are affine transformations of $G^{-1}$.
\end{remark}

The next proposition describes how variance aversion affects the extremal means, standard deviations, and values.

\begin{proposition}\label{pro:sensitivity}
Let $\mu_{F_{wo}}^\gamma$ and $\sigma_{F_{wo}}^\gamma$ (resp. $\mu_{F_{be}}^\gamma$ and $\sigma_{F_{be}}^\gamma$) denote the mean and standard deviation of the worst-case distribution $F_{wo}$ in Theorem \ref{thm:wo_quantile} (resp. the best-case distribution $F_{be}$ in Theorem \ref{thm:be_quantile}).
For two variance-aversion coefficients $\gamma_1 > \gamma_2 > 0$, the following comparative statics hold:
\begin{itemize}
\item[(i)] The optimal mean and standard deviation for problem \eqref{Prob:sup} satisfy $\mu_{F_{wo}}^{\gamma_1} < \mu_{F_{wo}}^{\gamma_2}$, $ \sigma_{F_{wo}}^{\gamma_1} > \sigma_{F_{wo}}^{\gamma_2}$ and
\begin{equation*}
\mu_{F_{wo}}^{\gamma_1} + \gamma_1 \left(\sigma_{F_{wo}}^{\gamma_1}\right)^2 > \mu_{F_{wo}}^{\gamma_2} + \gamma_2 \left(\sigma_{F_{wo}}^{\gamma_2}\right)^2.
\end{equation*}

\item[(ii)] The optimal mean and standard deviation for problem \eqref{Prob:inf} satisfy $\mu_{F_{be}}^{\gamma_1} > \mu_{F_{be}}^{\gamma_2}$, $ \sigma_{F_{be}}^{\gamma_1} < \sigma_{F_{be}}^{\gamma_2}$ and
\begin{equation*}
\mu_{F_{be}}^{\gamma_1} + \gamma_1 \left(\sigma_{F_{be}}^{\gamma_1}\right)^2 > \mu_{F_{be}}^{\gamma_2} + \gamma_2 \left(\sigma_{F_{be}}^{\gamma_2}\right)^2.
\end{equation*}
\end{itemize}
\end{proposition}
\begin{proof}
We only prove the case for the worst-case problem, as the proof for the best-case follows analogously.
For the worst-case problem, the optimal mean and standard deviation are parametrized by $\theta^* \in (0, \pi/2)$, which, for a given $\gamma > 0$, is the unique solution to $\tan(\theta^*) = 2 \gamma \sigma_G + 2 \gamma \sqrt{\epsilon} \sin(\theta^*)$ with $\mu_{F_{wo}}^\gamma = \mu_G + \sqrt{\epsilon} \cos(\theta^*)$ and $\sigma_{F_{wo}}^\gamma = \sigma_G + \sqrt{\epsilon} \sin(\theta^*)$.
Moreover, we have $\sin(\theta^*)(1 - 2 \gamma \sqrt{\epsilon} \cos(\theta^*)) = 2 \gamma \sigma_G \cos(\theta^*)$.
This implies $1 - 2 \gamma \sqrt{\epsilon} \cos(\theta^*) > 0$.
Note that
\begin{equation*}
\frac{d \theta^*}{d \gamma} = \frac{2 \sigma_G + 2 \sqrt{\epsilon} \sin(\theta^*)}{1 + \tan^2(\theta^*) - 2 \gamma \sqrt{\epsilon} \cos(\theta^*)} > 0.
\end{equation*}
Hence, $\theta^*$ is strictly increasing in $\gamma$.
Consequently,
\begin{equation*}
\frac{d \mu_{F_{wo}}^\gamma}{d \gamma} = - \sqrt{\epsilon} \sin(\theta^*) \frac{d \theta^*}{d \gamma} < 0, ~~ {\rm and} ~~ \frac{d \sigma_{F_{wo}}^\gamma}{d \gamma} = \sqrt{\epsilon} \cos(\theta^*) \frac{d \theta^*}{d \gamma} > 0.
\end{equation*}
Thus, for $\gamma_1 > \gamma_2$, we have $\mu_{F_{wo}}^{\gamma_1} < \mu_{F_{wo}}^{\gamma_2}$ and $ \sigma_{F_{wo}}^{\gamma_1} > \sigma_{F_{wo}}^{\gamma_2}$.
The inequality for the optimal value follows directly from the definition of the worst-case value:
\begin{equation*}
\mu_{F_{wo}}^{\gamma} + \gamma \left(\sigma_{F_{wo}}^{\gamma}\right)^2 = \sup_{F \in {\cal M}^2(G; \epsilon)} V^F(X).
\end{equation*}
This completes the proof.
\end{proof}

Proposition \ref{pro:sensitivity} highlights the trade-off between the mean and variance along the boundary of the Wasserstein ball. Along this boundary, a higher variance is necessarily associated with a lower mean, and vice versa. As the variance-aversion coefficient increases, variance receives greater weight in the objective, leading to a higher worst-case variance and a lower corresponding mean. In the best-case setting, these comparative statics are reversed.


From a geometric perspective, we seek the extremal values of $z:=\mu_F+\gamma\sigma_F^2$ subject to $(\mu_F-\mu_G)^2+(\sigma_F-\sigma_G)^2=\epsilon$. The solution to the problem \eqref{Prob:sup} is thereby characterized by the point of tangency between the curve $\sigma_F=\left(\frac{z-\mu_F}{\gamma}\right)^{\frac{1}{2}}$ and the circle $(\mu_F-\mu_G)^2+(\sigma_F-\sigma_G)^2=\epsilon$. When $\gamma$ becomes larger, the curve $\sigma_F=\left(\frac{z-\mu_F}{\gamma}\right)^{\frac{1}{2}}$ becomes flatter. Consequently, the point of tangency shifts in a manner that directly gives rise to the result in Proposition \ref{pro:sensitivity}(i) (see Fig. \ref{fig:impact_gamma}). An analogous geometric interpretation also applies to Proposition \ref{pro:sensitivity}(ii). 

\begin{figure}[!ht]
    \centering
    \includegraphics[width=0.5\linewidth]{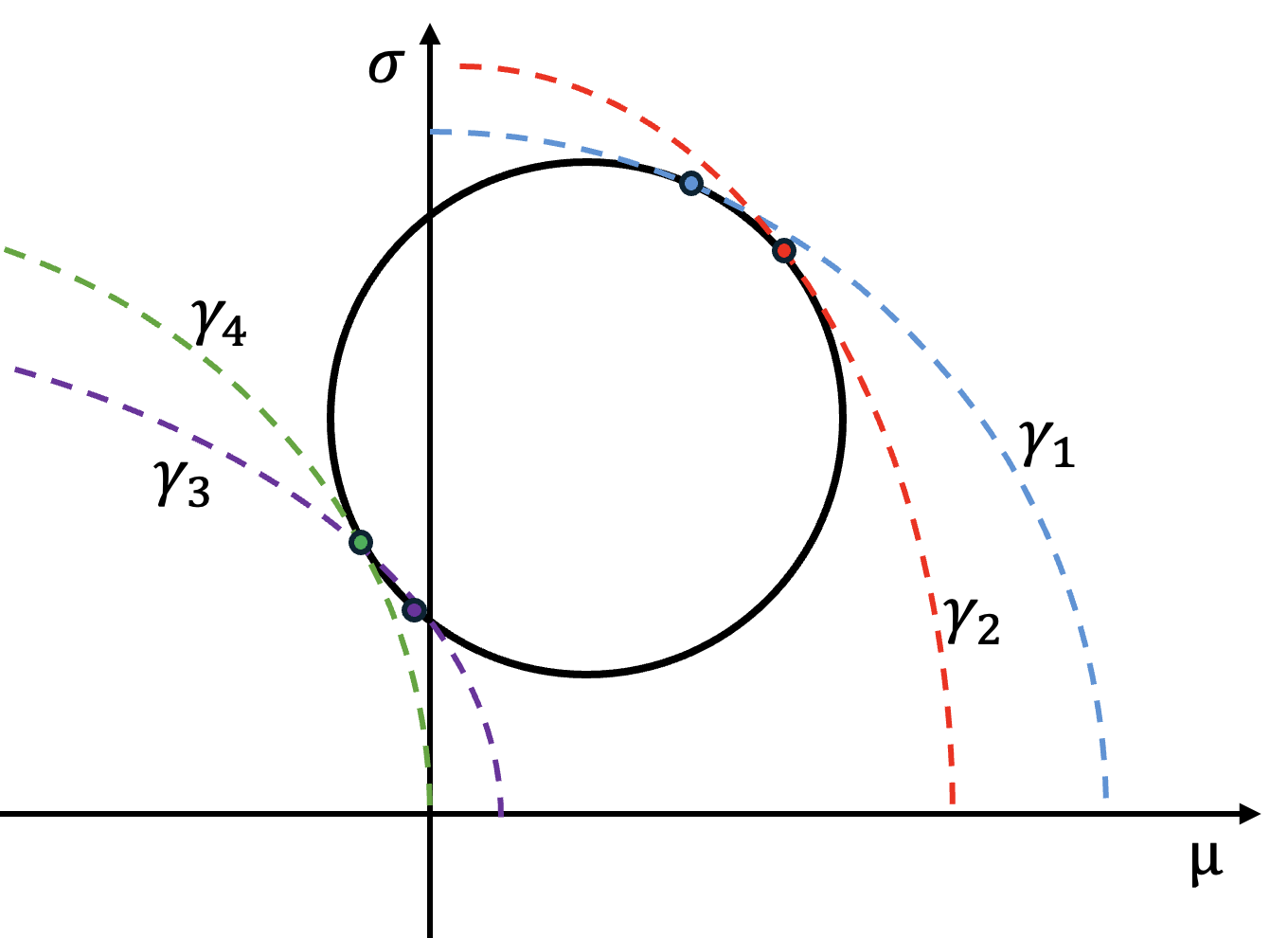}
    \caption{The impact of $\gamma$ on the worst- and best-case means and standard deviations.}
    \label{fig:impact_gamma}
\end{figure}

\section{Applications in risk sharing}\label{sec:alpha-maxmin-riskshare}

We now turn to risk sharing under distributional ambiguity. We consider $n$ agents, each of whom is assigned a share of the aggregate loss $X$. We assume that $X$ has a finite second moment under any probability measure considered in this section. Each agent $i\in\{1,\ldots,n\}$ evaluates her allocated loss $X_i$ using the mean-variance functional
\begin{equation}\label{MV_share}
V_i^{F_i}(\cdot;\gamma_i) = \mathbb E[\cdot] + \gamma_i \mathbb V[\cdot],
\end{equation}
where, in our context of risk sharing, $F_i$ denotes the CDF of $X$ under agent $i$'s belief, and $\gamma_i$ denotes the variance-aversion parameters of agent $i$. We suppress the subscripts of $\E[\cdot]$ and $\V[\cdot]$ to avoid the possible notational confusion when applying these operators to the allocated loss $X_i$.


We apply the results established in the preceding sections to two risk sharing problems. Subsection~\ref{sec:proportional-sharing} considers proportional risk sharing under heterogeneous beliefs. Subsection~\ref{sec:optimal-sharing} examines an $\alpha$-maxmin risk sharing problem under homogeneous beliefs.

\subsection{Proportional risk sharing with heterogeneous beliefs}\label{sec:proportional-sharing}

In this section, each agent is assumed to bear a proportional share of $X$. We write $\bm w=(w_1,\ldots,$ $w_n)^\top \in \Delta_n$ for the vector of sharing proportions. This constraint ensures that the aggregate loss is fully allocated among the agents, with each agent receiving a nonnegative fraction of the total exposure. Agent $i$ evaluates her share $w_iX$ using \eqref{MV_share} and describes uncertainty about $X$ by the Wasserstein ball $\mathcal M^2(G_i;\epsilon_i)$ centered at her reference distribution $G_i$. Let $\mu_i$ and $\sigma_i^2$ denote the mean and variance of $X$ under $G_i$. Besides that, we assume that all agents are extremely ambiguity averse and thus would only consider their choices under the worst-case scenarios.

A social planner chooses the proportional shares $\bm w^\top \in \Delta_n$ to minimize the sum of the individual objectives:
\begin{equation}\label{prob:hetero-rs}
\inf_{\bm w \in \Delta_n} \left\{\sum_{i=1}^n \sup_{F_i \in {\cal M}^2(G_i; \epsilon_i)} V^{F_i}_i(w_i X ; \gamma_i) \right\}.
\end{equation}
Applying the dual formulations \eqref{wo_values2} and \eqref{be_values2} from Theorems~ \ref{thm:wo_quantile} and~\ref{thm:be_quantile}, we have the following Proposition. To ease the presentation, we denote $\bm \lambda = (\lambda _1, \ldots, \lambda _n)$ and $\bm \gamma = (\gamma_1, \ldots, \gamma_n)$.

\begin{proposition}
For the problem \eqref{prob:hetero-rs}, we have
\begin{equation}\label{worst_prop}
\inf_{\bm w \in \Delta_n} \!\left\{\!\sum_{i=1}^n \sup_{F_i \in {\cal M}^2(G_i; \epsilon_i)} V^{F_i}_i(w_i X;\gamma_i) \right\} \!=\! \inf_{\bm \lambda > \bm \gamma} \inf_{\bm w \in \Delta_n} \sum_{i=1}^n \left\{\left(\frac{\lambda_i \gamma_i}{\lambda_i - \gamma_i} \sigma_i^2 + \lambda_i \epsilon_i\right) w_i^2 + \mu_i w_i + \frac{1}{4 \lambda_i} \right\}.
\end{equation}
Furthermore, if all agents agree on the mean of $X$ under their reference distributions, i.e., $\mu_1=\cdots=\mu_n=\mu$, the optimal proportions that solve the problem \eqref{prob:hetero-rs} are
\begin{equation}
    w_i^*=\frac{\frac{1}{\frac{\lambda_i^*\gamma_i}{\lambda_i^*-\gamma_i}\sigma_i^2+\lambda_i^*\epsilon_i}}{\sum_{k=1}^n\frac{1}{\frac{\lambda_k^*\gamma_k}{\lambda_k^*-\gamma_k}\sigma_k^2+\lambda_k^*\epsilon_k}},\quad i=1,\dots,n,
\end{equation}
where $\bm\lambda^*$ is given by
\begin{equation}
    \bm\lambda^*=\arginf_{\bm\lambda>\bm\gamma}\ \left\{\left(\sum_{k=1}^n\frac{1}{\frac{\lambda_k\gamma_k}{\lambda_k-\gamma_k}\sigma_k^2+\lambda_k\epsilon_k}\right)^{-1}+\sum_{k=1}^n\frac{1}{4\lambda_k}\right\}.
\end{equation}
\end{proposition}
\begin{proof}
For each agent $i$, note that
$$
V_i^{F_i}(w_i X;\gamma_i)=w_i\E[X_i]+w_i^2\gamma_i\V[X_i]=w_i V_i^{F_i}(X;w_i\gamma_i).
$$
Applying the dual formulation embedded in the proof of Theorem~\ref{thm:wo_quantile} then yields
\begin{align*}
&\inf_{\bm w \in \Delta_n} \left\{\sum_{i=1}^n \sup_{F_i \in {\cal M}^2(G_i; \epsilon_i)} V^{F_i}_i(w_i X;\gamma_i) \right\} \\
=&{\inf_{\bm w \in \Delta_n} \left\{\sum_{i=1}^n \inf_{\tilde{\lambda}_i > w_i\gamma_i} \left\{w_i \mu_i + \frac{\tilde{\lambda_i}w_i^2 \gamma_i}{\tilde{\lambda}_i - w_i\gamma_i}\sigma_i^2 + \frac{w_i}{4 \tilde{\lambda_i}}  + \tilde{\lambda}_i w_i \epsilon_i \right\}\right\}}  \\
=& \inf_{\bm w \in \Delta_n} \left\{\sum_{i=1}^n \inf_{\lambda_i > \gamma_i} \left\{w_i \mu_i + \frac{\lambda_i \gamma_i}{\lambda_i - \gamma_i} w_i^2 \sigma_i^2 + \frac{1}{4 \lambda_i}  + \lambda_i w_i^2 \epsilon_i \right\}\right\}  \\
=& \inf_{\bm w \in \Delta_n} \inf_{\bm \lambda > \bm \gamma} \sum_{i=1}^n \left\{\left(\frac{\lambda_i \gamma_i}{\lambda_i - \gamma_i} \sigma_i^2 + \lambda_i \epsilon_i\right) w_i^2 + \mu_i w_i + \frac{1}{4 \lambda_i} \right\}  \\
=& \inf_{\bm \lambda > \bm \gamma} \inf_{\bm w \in \Delta_n} \sum_{i=1}^n \left\{\left(\frac{\lambda_i \gamma_i}{\lambda_i - \gamma_i} \sigma_i^2 + \lambda_i \epsilon_i\right) w_i^2 + \mu_i w_i + \frac{1}{4 \lambda_i} \right\},
\end{align*}
where the second equality follows from the change of variables $\tilde{\lambda}_i = w_i \lambda_i$ when $w_i>0$, while for $w_i=0$ the corresponding infimum is obtained in the limit $\lambda_i \to \infty$. When all agents agree on the mean of $X$ under their respective reference distributions, we further have
\begin{align}
    &\inf_{\bm \lambda > \bm \gamma} \inf_{\bm w \in \Delta_n} \sum_{i=1}^n \left\{\left(\frac{\lambda_i \gamma_i}{\lambda_i - \gamma_i} \sigma_i^2 + \lambda_i \epsilon_i\right) w_i^2 + \mu_i w_i + \frac{1}{4 \lambda_i} \right\} \nonumber \\
    =\ &\inf_{\bm\lambda>\bm\gamma}\left\{\inf_{\bm w\in\Delta_n}\sum_{i=1}^n\left\{\left(\frac{\lambda_i\gamma_i}{\lambda_i-\gamma_i}\sigma_i^2+\lambda_i\epsilon_i\right)w_i^2+\frac{1}{4\lambda_i}\right\}+\mu\right\}, \label{reduced1}
\end{align}
where the inner problem is a constrained quadratic program whose solution is
$$
w_i^*=\frac{\frac{1}{\frac{\lambda_i\gamma_i}{\lambda_i-\gamma_i}\sigma_i^2+\lambda_i\epsilon_i}}{\sum_{k=1}^n\frac{1}{\frac{\lambda_k\gamma_k}{\lambda_k-\gamma_k}\sigma_k^2+\lambda_k\epsilon_k}},\quad i=1,\dots,n,
$$
for fixed $\bm\lambda>\bm\gamma$. Substituting these shares into \eqref{reduced1} gives
$$
\inf_{\bm\lambda>\bm\gamma} \left\{\left(\sum_{k=1}^n\frac{1}{\frac{\lambda_k\gamma_k}{\lambda_k-\gamma_k}\sigma_k^2+\lambda_k\epsilon_k}\right)^{-1}+\sum_{k=1}^n\frac{1}{4\lambda_k}\right\}.
$$
This completes the proof.
\end{proof}

\begin{remark}\label{remark_worst}
For the general case, where there is no agreement on the mean of $X$ under different beliefs, the optimal shares can still be characterized via the Karush-Kuhn-Tucker (KKT) conditions \citep[see, e.g., Corollary 1 of Section 8.3 and Theorem 2 of Section 8.4 of][]{luenberger1997optimization}. 
Define
\begin{equation*}
f_i(w_i) := \left(\frac{\lambda_i \gamma_i}{\lambda_i - \gamma_i} \sigma_i^2 + \lambda_i \epsilon_i\right) w_i^2 + \mu_i w_i + \frac{1}{4 \lambda_i}.
\end{equation*}
For fixed $\bm \lambda > \bm \gamma$, the optimal $\bm w^*$ minimizes $\sum_{i=1}^n f_i(w_i)$ subject to $\sum_{i=1}^n w_i=1$ and $w_i \geq 0$. 
The KKT conditions yield
\begin{equation*}
2 \left(\frac{\lambda_i \gamma_i}{\lambda_i - \gamma_i} \sigma_i^2 + \lambda_i \epsilon_i\right) w_i^* + \mu_i - \nu - \eta_i = 0,~~ i = 1, \ldots, n,
\end{equation*}
where $\nu \in \R$ is the Lagrange multiplier for the equality constraint $\sum_{i=1}^n w_i=1$, and $\eta_i \geq 0$ is the multiplier for $w_i \geq 0$, with complementarity $\eta_i w_i^*=0$.
For the active set $\mathcal A := \{i \in \{1, \ldots, n\}: w_i^* > 0\}$, we have
\begin{equation*}
w_i^* = \frac{\nu - \mu_i}{2 \left(\frac{\lambda_i \gamma_i}{\lambda_i - \gamma_i} \sigma_i^2 + \lambda_i \epsilon_i\right)},~~ i \in {\cal A},
\end{equation*}
and $\sum_{i \in {\cal A}} w_i^*=1$, $\nu - \mu_i > 0$ for $i \in {\cal A}$.
\end{remark}

The following example considers two agents who agree on the reference mean of $X$.

\begin{example}
Consider $n=2$ and $\mu_1 = \mu_2$. For fixed dual variables $\lambda_1$ and $\lambda_2$, the sharing proportions that minimize the inner problem in \eqref{worst_prop} are
\begin{equation*}
w_1 = \frac{\frac{\lambda_2 \gamma_2}{\lambda_2 - \gamma_2} \sigma_2^2 + \lambda_2 \epsilon_2}{\frac{\lambda_1 \gamma_1}{\lambda_1 - \gamma_1} \sigma_1^2 + \lambda_1 \epsilon_1 + \frac{\lambda_2 \gamma_2}{\lambda_2 - \gamma_2} \sigma_2^2 + \lambda_2 \epsilon_2}, ~~ w_2 = 1 - w_1.
\end{equation*}
The optimal $\lambda_1$ and $\lambda_2$ solve the outer minimization in \eqref{worst_prop}.
For given dual variables, agent 1 bears the larger share precisely when
\begin{equation*}
\frac{\lambda_1\gamma_1}{\lambda_1-\gamma_1}\sigma_1^2+\lambda_1\epsilon_1
<
\frac{\lambda_2\gamma_2}{\lambda_2-\gamma_2}\sigma_2^2+\lambda_2\epsilon_2.
\end{equation*}
In particular, if $\sigma_1^2=\sigma_2^2$, $\epsilon_1=\epsilon_2$, and the two dual variables coincide, then $\gamma_1<\gamma_2$ implies $w_1>w_2$.
Thus, in this special case, the more variance-averse agent bears the smaller share.
\end{example}

\subsection{{$\alpha$-maxmin mean-variance risk sharing}}\label{sec:optimal-sharing}

A social planner allocates the aggregate loss $X$ among $n$ agents without restricting the form of the allocation.
Under homogeneous beliefs, all agents use the same probability measure $\Prob$ on $(\Omega,\mathcal F)$. Let $F$ denote the CDF of the aggregate loss $X$ under $\Prob$, with finite mean $\mu_F$ and variance $\sigma_F^2$, both assumed to be known. For an allocation $\{X_i\}_{i=1}^n \in (\mathcal L^2)^n$ satisfying $\sum_{i=1}^n X_i=X$, agent $i$ evaluates her allocated loss through the mean-variance functional $V_i^F(X_i)$ defined in \eqref{MV_share}.
The planner assigns equal welfare weights to the agents.

Uncertainty about the distribution of $X$ is described by the Wasserstein ball $\mathcal M^2(G;\epsilon)$.
The parameter $\alpha\in[0,1]$ is the weight placed on the worst-case evaluation in the $\alpha$-maxmin criterion below.
The planner is ambiguity-seeking for $\alpha\in[0,1/2)$, ambiguity-neutral at $\alpha=1/2$, and ambiguity-averse for $\alpha\in(1/2,1]$; ambiguity aversion increases with $\alpha$. The planner solves
\begin{equation}\label{Prob:risk-sharing}
    \inf_{\{X_i\}_{i=1}^n\in\mathcal{S}(X)}\left\{\alpha\sup_{F_w \in \mathcal M^2(G; \epsilon)} \sum_{i=1}^n V_i^{F_w}(X_i;\gamma_i)+(1-\alpha)\inf_{F_b \in \mathcal M^2(G; \epsilon)} \sum_{i=1}^n V_i^{F_b}(X_i;\gamma_i)\right\},
\end{equation}
where $\mathcal{S}(X)=\left\{\{X_i\}_{i=1}^n \in (\mathcal L^2)^n:\ \sum_{i=1}^nX_i=X\right\}$.
The distributions $F_w$ and $F_b$ belong to the same Wasserstein ball but need not coincide: $F_w$ gives the worst-case evaluation of aggregate loss, and $F_b$ gives the best-case evaluation.

To solve the problem \eqref{Prob:risk-sharing}, we first characterize the optimal allocation for a fixed probability measure.

\begin{lemma}\label{lem:weighted-variance}
Let $X$ have finite mean $\mu_F$ and variance $\sigma_F^2$.
Consider the problem
\begin{equation}\label{Prob:auxiliary}
    \inf_{\{X_i\}_{i=1}^n\in\mathcal{S}(X)}\ \sum_{i=1}^n V_i^F(X_i;\gamma_i).
\end{equation}
The set of optimal allocations $(X_1^*,\dots,X_n^*)$ is given by
\begin{equation*}
X_i^* = \frac{1/\gamma_i}{\sum_{j = 1}^n 1/\gamma_j} X + \beta_i,\qquad i = 1, \ldots, n,
\end{equation*}
where $\beta_1, \ldots, \beta_n \in \R$ satisfy $\sum_{i=1}^n \beta_i = 0$.
\end{lemma}
\begin{proof}
Set
\begin{equation*}
T:=\sum_{j=1}^n\frac{1}{\gamma_j}\qquad\text{and}\qquad
c_i:=\frac{1/\gamma_i}{T},
\qquad i=1,\ldots,n.
\end{equation*}
For an arbitrary feasible allocation, let $\bar X:=X-\E[X]$ and $\bar X_i:=X_i-\E[X_i]$.
Because $\sum_{i=1}^n\bar X_i=\bar X$, $\gamma_i c_i=1/T$, and $\sum_{i=1}^n\gamma_i c_i^2=1/T$, we have
\begin{align*}
\sum_{i=1}^n\gamma_i\E\left[(\bar X_i-c_i\bar X)^2\right]
&=\sum_{i=1}^n\gamma_i\V[X_i]
-\frac{2}{T}\sum_{i=1}^n\Cov(X_i,X)
+\frac{1}{T}\V[X]\\
&=\sum_{i=1}^n\gamma_i\V[X_i]-\frac{1}{T}\V[X]\geq0.
\end{align*}
Moreover, feasibility implies $\sum_{i=1}^n\E[X_i]=\E[X]$. Hence
\begin{equation*}
\sum_{i=1}^n V_i^F(X_i;\gamma_i)
\geq \E[X]+\frac{1}{T}\V[X].
\end{equation*}
Equality holds if and only if $\bar X_i=c_i\bar X$ almost surely for every $i$, or equivalently,
$X_i=c_iX+\beta_i$ for constants $\beta_i$ satisfying $\sum_{i=1}^n\beta_i=0$.
These constants represent deterministic transfers: they leave the variances unchanged, and their zero sum preserves the aggregate expected loss. This proves the claim.
\end{proof}

\begin{remark}
The proportional sharing rule in Lemma \ref{lem:weighted-variance} is the mean-variance instance of the general proportional allocation rule for dilated risk measures, as noted in Remark 4.6 of \cite{acciaio2007optimal} and proved in Theorem 3.9 of \cite{barrieu2005inf} within a convex risk measure framework. 
Lemma \ref{lem:weighted-variance} provides a self-contained proof of this rule specifically for the mean-variance criterion. The weighted-variance decomposition proves the result directly, without convex duality or the general theory of risk measures. We use the same allocation in the ambiguity model below.
\end{remark}

The following theorem fully solves the problem \eqref{Prob:risk-sharing}.

\begin{theorem}\label{thm:op_risk_sharing}
    An optimal allocated risks $X_1^*,\dots,X_n^*$ for the problem \eqref{Prob:risk-sharing} are the same as those shown in Lemma \ref{lem:weighted-variance}. The quantile functions of $X$ corresponding to the worst-case and best-case probability measures are shown in Theorems \ref{thm:wo_quantile} and \ref{thm:be_quantile} with $\gamma=\frac{1}{\sum_{i=1}^n 1/\gamma_i}$. 
\end{theorem}
\begin{proof}
As shown by Lemma \ref{lem:weighted-variance}, the solution $(X_1^*, \ldots, X_n^*)$ to the problem \eqref{Prob:auxiliary} does not depend on $\Prob$, i.e., $X_i^* = \frac{1/\gamma_i}{\sum_{j = 1}^n 1/\gamma_j} X,~ i = 1, \ldots, n$. 
Therefore, for any given $F_b$, we have
\begin{align*}
&\sup_{F_w \in \mathcal M^2(G; \epsilon)}\left\{\alpha\sum_{i=1}^n V_i^{F_w}(X_i^*;\gamma_i) + (1-\alpha)\sum_{i=1}^n V_i^{F_b}(X_i^*;\gamma_i)\right\} \\
\le\ &\sup_{F_w \in \mathcal M^2(G; \epsilon)}\left\{\alpha\sum_{i=1}^n V_i^{F_w}(X_i;\gamma_i) + (1-\alpha)\sum_{i=1}^n V_i^{F_b}(X_i;\gamma_i)\right\}
\end{align*}
for any $\{X_i\}_{i=1}^n\in\mathcal{S}(X)$, where the equality holds when $\{X_i\}_{i=1}^n=\{X_i^*\}_{i=1}^n$.
This leads to 
\begin{align*}
    &\sup_{F_w \in \mathcal M^2(G; \epsilon)}\inf_{\{X_i\}_{i=1}^n\in\mathcal{S}(X)}\left\{\alpha\sum_{i=1}^n V_i^{F_w}(X_i;\gamma_i)+(1-\alpha)\sum_{i=1}^n V_i^{F_b}(X_i;\gamma_i)\right\} \\
    =&\inf_{\{X_i\}_{i=1}^n\in\mathcal{S}(X)}\sup_{F_w \in \mathcal M^2(G; \epsilon)}\left\{\alpha\sum_{i=1}^n V_i^{F_w}(X_i;\gamma_i)+(1-\alpha)\sum_{i=1}^n V_i^{F_b}(X_i;\gamma_i)\right\}.
\end{align*}
As such,
\begin{align}
    &\inf_{\{X_i\}_{i=1}^n\in\mathcal{S}(X)}\left\{\alpha\sup_{F_w \in \mathcal M^2(G; \epsilon)}\sum_{i=1}^n V_i^{F_w}(X_i;\gamma_i)+(1-\alpha)\inf_{F_b \in \mathcal M^2(G; \epsilon)}\sum_{i=1}^n V_i^{F_b}(X_i;\gamma_i)\right\} \nonumber \\
    =& \inf_{\{X_i\}_{i=1}^n\in\mathcal{S}(X)}\inf_{F_b \in \mathcal M^2(G; \epsilon)}\sup_{F_w \in \mathcal M^2(G; \epsilon)}\left\{\alpha\sum_{i=1}^n V_i^{F_w}(X_i;\gamma_i)+(1-\alpha)\sum_{i=1}^n V_i^{F_b}(X_i;\gamma_i)\right\} \nonumber \\
    =& \inf_{F_b \in \mathcal M^2(G; \epsilon)}\inf_{\{X_i\}_{i=1}^n\in\mathcal{S}(X)}\sup_{F_w \in \mathcal M^2(G; \epsilon)}\left\{\alpha\sum_{i=1}^n V_i^{F_w}(X_i;\gamma_i)+(1-\alpha)\sum_{i=1}^n V_i^{F_b}(X_i;\gamma_i)\right\} \nonumber \\
    =& \inf_{F_b \in \mathcal M^2(G; \epsilon)}\sup_{F_w \in \mathcal M^2(G; \epsilon)}\left\{\alpha\sum_{i=1}^n V_i^{F_w}(X_i^*;\gamma_i)+(1-\alpha)\sum_{i=1}^n V_i^{F_b}(X_i^*;\gamma_i)\right\}. \label{Prob:reduced}
\end{align}
Note that
$$
\sum_{i=1}^n V_i^F(X_i^*;\gamma_i)= \mu_{F}+\frac{\sigma_{F}^2}{\sum_{i=1}^n 1/\gamma_i}.
$$
The problem \eqref{Prob:reduced} becomes
\begin{align*}
    &\inf_{F_b \in \mathcal M^2(G; \epsilon)}\sup_{F_w \in \mathcal M^2(G; \epsilon)}\left\{\alpha\sum_{i=1}^n V_i^{F_w}(X_i^*;\gamma_i)+(1-\alpha)\sum_{i=1}^n V_i^{F_b}(X_i^*;\gamma_i)\right\} \\
    =& \sup_{F_w\in\mathcal{M}^2(G;\epsilon)}\alpha\left(\mu_{F_w}+\frac{\sigma_{F_w}^2}{\sum_{i=1}^n 1/\gamma_i}\right)+\inf_{F_b\in\mathcal{M}^2(G;\epsilon)}(1-\alpha)\left(\mu_{F_b}+\frac{\sigma_{F_b}^2}{\sum_{i=1}^n 1/\gamma_i}\right).
\end{align*}
The worst-case and best-case CDFs of $X$ can be characterized via their quantile functions as shown in Theorems \ref{thm:wo_quantile} and \ref{thm:be_quantile}. The proof is complete.
\end{proof}

Lemma \ref{lem:weighted-variance} and Theorem \ref{thm:op_risk_sharing} show that the optimal sharing rule depends on the agents' risk attitudes only, but not on distributional ambiguity or $\alpha$. The corresponding extremal distributions are given by Theorems \ref{thm:wo_quantile} and \ref{thm:be_quantile} with $\gamma$ therein being $(\sum_{i=0}^n 1/\gamma_i)^{-1}$.

\begin{remark}
The results established in this section also apply to the optimal insurance problem with the variance premium principle. Let $I_0(X) = X - \sum_{i=1}^n I_i(X)$ denote the retained loss, and let $I_i(X)$ denote the loss ceded to insurer $i$. Under the variance premium principle, the mean-variance objective reduces to
\begin{equation*}
\E[X] + \sum_{i=0}^n \gamma_i \V[I_i(X)],
\end{equation*}
where $\gamma_0$ is the insured's variance-aversion coefficient and $\gamma_i~(i=1, \ldots, n)$ are the insurers' premium loadings.
Lemma \ref{lem:weighted-variance} and Theorem \ref{thm:op_risk_sharing} show that the optimal retained loss and indemnity functions are
\begin{equation*}
I_i^*(X) = \frac{1/\gamma_i}{\sum_{j=0}^n 1/\gamma_j} X, \quad i = 0, \ldots, n,
\end{equation*}
and the extremal quantile functions for $X$ are given by Theorems \ref{thm:wo_quantile} and \ref{thm:be_quantile} with $\gamma$ therein being $(\sum_{i=0}^n 1/\gamma_i)^{-1}$. As in Theorem \ref{thm:op_risk_sharing}, the optimal contract is independent of both the ambiguity-set parameter $\epsilon$ and the ambiguity-attitude parameter $\alpha$.
\end{remark}


\section{Extension to distortion riskmetrics}\label{sec:extension}

Distortion riskmetrics, introduced by \citet{wang2020distortion}, include monetary risk measures, deviation measures, and many other functionals used in finance and insurance. For a random variable $Z$ with CDF $F$, its distortion riskmetric is defined by
\begin{equation}\label{eq:DRM}
    \rho_g(Z)=\int_{0}^\infty g(\Prob(Z>z))dz+\int_{-\infty}^0\left(g(\Prob(Z>z))-g(1)\right)dz,
\end{equation}
where $g$ is a \emph{distortion function} belonging to
\begin{equation*}
    \mathcal{G}=\left\{g:[0,1]\mapsto \R:\ g\ \text{is continuous and of bounded variation and}\ g(0)=0\right\}.
\end{equation*}
Unlike the classical distortion risk measure \citep{yaari1987dual}, which requires $g$ to be increasing and to satisfy $g(1)=1$, the distortion riskmetric imposes neither restriction. The class therefore includes coherent risk measures such as Tail Value-at-Risk (TVaR), as well as RVaR, the Gini deviation ($g_G(t)=t-t^2$), the mean-median deviation ($g_M(t)=t\wedge (1-t)$), and riskmetrics generated by continuous inverse-S-shaped distortions.

By replacing the expectation term with a distortion riskmetric in  \eqref{Prob:sup} and \eqref{Prob:inf}, we have the following extended optimization models:
\begin{equation}\label{Prob:sup2}
    \sup_{F\in\mathcal{M}^2(G;\epsilon)}\ \rho_g^F(X)+\gamma\V_F[X]
\end{equation}
and
\begin{equation}\label{Prob:inf2}
    \inf_{F\in\mathcal{M}^2(G;\epsilon)}\ \rho_g^F(X)+\gamma\V_F[X].
\end{equation}
The distortion term can distinguish between distributions with the same mean and variance, so the extremal quantile need not remain a location-scale transformation of the reference quantile.

We present the following two examples to show that a variety of distributionally robust constrained optimization problems can be accommodated within our framework.

\begin{example}
    Existing studies of extremal distortion risk measures and riskmetrics impose both moment information and a 2-Wasserstein constraint \citep{bernard2024robust,liu2025robust}. When the first two moments must be estimated from finite samples, one may instead consider
    \begin{equation}\label{Prob:extension1}
    \begin{aligned}
        \sup_{F\in\mathcal{M}^2(G;\epsilon)}&\ \rho_{g_0}^F(X), \\
        \text{s.t.}&\ \E_F[X]\le\Bar{\mu},\quad \V_F[X]\le\Bar{\sigma}^2,
    \end{aligned}
    \end{equation}
    with $g_0\in\mathcal G$. By Lemma \ref{lem:quantile-formulation}, the objective is linear in $F^{-1}$ and both constraints are convex in $F^{-1}$. Strict feasibility of the reference distribution $G$ ensures Slater's condition, so \eqref{Prob:extension1} has the same value as its Lagrange dual:
    \begin{equation*}
        \inf_{\lambda_1\ge 0,\lambda_2\ge 0}\sup_{F\in\mathcal{M}^2(G;\epsilon)}\ \rho_{g_0}^F(X)-\lambda_1\E_F[X]-\lambda_2\V_F[X]+\lambda_1\Bar{\mu}+\lambda_2\Bar{\sigma}^2,
    \end{equation*}
    for which the inner problem is equivalent to 
    \begin{align*}
        \sup_{F\in\mathcal{M}^2(G;\epsilon)}\ \rho_{g_0}^F(X)-\lambda_1\E_F[X]-\lambda_2\V_F[X]=&-\inf_{F\in\mathcal{M}^2(G;\epsilon)}\rho_{-g_0}^F(X)+\lambda_1\E_F[X]+\lambda_2\V_F[X] \\
        =&-\inf_{F\in\mathcal{M}^2(G;\epsilon)}\rho_g^F(X)+\lambda_2\V_F[X],
    \end{align*}
    where $g(t):=-g_0(t)+\lambda_1t\in\mathcal{G}$. This is exactly the form of \eqref{Prob:inf2}.
\end{example}

\begin{example}
    Following \citet{chen2019distributionally}, we may express partial distributional information through a stochastic-order constraint. Suppose $X\leq_{\rm icx}Y$, where $Y$ takes the values $0$ and $\eta$ with probabilities $\alpha$ and $1-\alpha$, respectively, for some $\eta>0$. By the equivalence between increasing convex order and stop-loss order, for every $t\in[0,\eta]$,
    \begin{align*}
        \E_F[(X-t)_+]\le \E_F[(Y-t)_+]\ &\Longleftrightarrow\ \E_F[(X-t)_+]\le (1-\alpha)(\eta-t) \\
        &\Longleftrightarrow\ \frac{1}{1-\alpha}\E[(X-t)_+]+t\le\eta.
    \end{align*}
    It follows that $\inf_{t\in\R}\left\{\frac{1}{1-\alpha}\E_F[(X-t)_+]+t\right\}\leq\eta$, or equivalently, $\mathrm{TVaR}_\alpha^F(X)\leq\eta$. Adding this constraint to \eqref{Prob:sup} gives
    \begin{equation*}
    \begin{aligned}
        \sup_{F\in\mathcal{M}^2(G;\epsilon)}&\ \E_F[X]+\gamma\V_F[X], \\
        \text{s.t.}&\ \rho_{g_{T,\alpha}}^F(X)\le\eta.
    \end{aligned}
    \end{equation*}
    Since $\mathrm{TVaR}$ is linear in $F^{-1}$, the Lagrange dual is
    \begin{equation*}
        \inf_{\lambda\ge 0}\sup_{F\in\mathcal{M}^2(G;\epsilon)}\ \E_F[X]+\gamma\V_F[X]-\lambda \rho_{g_{T,\alpha}}^F(X)+\lambda\eta,
    \end{equation*}
    whose inner optimization has the form \eqref{Prob:sup2} with $g(t)=t-\lambda g_{T,\alpha}(t)$.
\end{example}

Similar to Lemma \ref{lem:binding}, the following lemma shows that the extreme values of $\rho_g^F(X)+\gamma\V_F[X]$ can be attained at the boundary of the Wasserstein ball. Although the proof is very similar to that of Lemma \ref{lem:binding}, we include it here for completeness.
\begin{lemma}\label{lem:binding2}
     Any optimizer $F^*$ for the worst-case problem \eqref{Prob:sup2} (resp. best-case problem \eqref{Prob:inf2}) satisfies $W_2(F^*,G) = \sqrt{\epsilon}$. Moreover, if $g(1)\neq 0$, then $F^*$ must satisfy $W_2(F^*,G) = \sqrt{\epsilon}$.
\end{lemma}
\begin{proof}
    Since $g$ is a continuous distortion function, by Lemma \ref{lem:quantile-formulation} we have 
    $$
    \rho_g^F(X)=\int_0^1F^{-1}(1-t)dg(t)=\int_0^1F^{-1}(t)d\Tilde{g}(t),
    $$
    where $\Tilde{g}(t)=g(1)-g(1-t)$ is the dual distortion function of $g$. 
    
    Let $F$ be a CDF that lies in the interior of the Wasserstein ball, i.e., 
	$$\int_0^1|F^{-1}(t)-G^{-1}(t)|^2dt<\epsilon.$$ 
	Since $F$ lies in the interior of the Wasserstein ball, we may choose $c\in\R$ such that $cg(1)\ge 0$ and
	$$\int_0^1|F^{-1}(t)+c-G^{-1}(t)|^2dt=\epsilon.$$ 
	Define $\Tilde{F}$ by $\Tilde{F}^{-1}(t)=F^{-1}(t)+c$, then
    \begin{align*}
        &\rho_{g}^{\Tilde{F}}(X)+\gamma\V_{\Tilde{F}}[X] \\
        =&\int_0^1(F^{-1}(t)+c)d\Tilde{g}(t)+\gamma\left(\int_0^1(F^{-1}(t))^2dt-\left(\int_0^1F^{-1}(t)dt\right)^2\right) \\
        =&\rho_g^F(X)+\gamma\V_F[X]+c(\Tilde{g}(1)-\Tilde{g}(0)) \\
        =&\rho_g^F(X)+\gamma\V_F[X]+cg(1)\ge \rho_g^F(X)+\gamma\V_F[X],
    \end{align*}
    where the last inequality is strict if $g(1)\neq 0$. This ends the proof.
\end{proof}

For any $g\in\mathcal{G}$, define its convex envelope by
\begin{equation*}
    g_*=\sup\{k\in\mathcal{G}: k\ \text{is convex on $[0,1]$ and}\ k\le g\}.
\end{equation*}
We also define the functions 
$$h^\lambda(t)=\Tilde{g}(t)+2(\gamma+\lambda)\int_0^tG^{-1}(s)ds, ~~~ \Lambda(\lambda)=\left(\int_0^1((h^\lambda_*)'(t)-h^\lambda(1))^2dt\right)^{\frac{1}{2}}.$$
It is straightforward to verify that $h^\lambda\in\mathcal{G}$. 
In the following, $\tilde{g}_*'$ and $(h^\lambda_*)'$ are understood as the left-hand derivatives of $\tilde{g}_*$ and $h^\lambda_*$. To proceed, we impose the following technical assumption.
\begin{assumption}\label{ass:nonconstant}
    $\int_0^1({\tilde{g}_*}'(t))^2dt<\infty$ and the derivative $(h^\lambda_*)'(t)$ is not a constant over $t\in(0,1)$ for any $\lambda\ge 0$.
\end{assumption}
Note that if $g$ is absolutely continuous with $\int_0^1(g'(t))^2dt<\infty$, then the square integrability of the derivative is preserved by $\tilde{g}_*'$. The assumption of $(h_*^\lambda)'$ being non-constant aligns with Assumption 3.6 of \cite{bernard2024robust} and Assumption A of \cite{liu2025robust} and guarantees the attainability of the constructed worst-case distribution in Theorem \ref{thm:wo_quantile2}.  
The following lemma shows the continuity of $\Lambda(\lambda)$ over $(0,\infty)$ under Assumption \ref{ass:nonconstant}.
\begin{lemma}\label{lem:continuity_Lambda}
    Let Assumption \ref{ass:nonconstant} hold. The function $\Lambda(\lambda)$ is continuous on $(0,\infty)$.
\end{lemma}
\begin{proof}
    It is obvious that both $h^\lambda(t)$ and $h^\lambda_*(t)$ are increasing in $\lambda$. Note that $h^\lambda_*(t)$ admits the representation 
    $$
    h^\lambda_*(t)=\inf\big\{\beta h^\lambda(t_1)+(1-\beta)h^\lambda(t_2): \beta\ge 0,\ \beta t_1+(1-\beta)t_2=t\big\}.
    $$
    Since $h^\lambda$ is continuous, for any $t\in[0,1]$, there exist $\beta_{\lambda,t}\in[0,1]$ and $t_1\le t\le t_2$ such that $h^\lambda_*(t)=\beta_{\lambda,t}h^\lambda(t_1)+(1-\beta_{\lambda,t})h^\lambda(t_2)$. As such, 
    \begin{align*}
        h^{\lambda_1}_*(t)=\beta_{\lambda_1,t} h^{\lambda_1}(t_1)+(1-\beta_{\lambda_1,t})h^{\lambda_1}(t_2)\le h^{\lambda_2}_*(t)\le \beta_{\lambda_1,t}h^{\lambda_2}(t_1)+(1-\beta_{\lambda_1,t})h^{\lambda_2}(t_2)
    \end{align*}
    when $\lambda_1\le \lambda_2$, which results in
    \begin{align*}
        &|h^{\lambda_1}_*(t)-h^{\lambda_2}_*(t)|\le |\beta_{\lambda_1,t}(h^{\lambda_1}(t_1)-h^{\lambda_2}(t_1))+(1-\beta_{\lambda_1,t})(h^{\lambda_1}(t_2)-h^{\lambda_2}(t_2))| \\
        \Longrightarrow\ &|h^{\lambda_1}_*(t)-h^{\lambda_2}_*(t)|\le |\lambda_1-\lambda_2|\beta_{\lambda_1,t} \left|\int_0^{t_1}G^{-1}(s)ds \right| +|\lambda_1-\lambda_2|(1-\beta_{\lambda_1,t}) \left|\int_0^{t_2}G^{-1}(s)ds \right| \\
        \Longrightarrow\ &|h^{\lambda_1}_*(t)-h^{\lambda_2}_*(t)|\le |\lambda_1-\lambda_2| \int_0^1 \left| G^{-1}(s) \right| ds.
    \end{align*}
    Therefore, $\{h^{\lambda_n}_*(t)\}$ uniformly converges to $h^{\lambda}_*(t)$ when $\lambda_n\to\lambda$. By Theorem 25.7 of \cite{rockafellar1970convex}, we get that $\{{h^{\lambda_n}_*}'(t)\}$ converges to ${h^\lambda_*}'(t)$ on the set where $h^\lambda_*$ is differentiable. Since $h^\lambda_*$ is differentiable almost everywhere on $[0,1]$, we further get that $\{{h^{\lambda_n}_*}'(t)\}$ converges to ${h^\lambda_*}'(t)$ almost everywhere on $[0,1]$.

    Now let $Z_1={h^\lambda_*}'(U_1)$ and $Z_2={\tilde{g}_*}'(U_2)+2(\gamma+\lambda)G^{-1}(U_2)$, where $U_1$ and $U_2$ are independent uniform random variables on $(0,1)$. It is easy to compute that
    \begin{align*}
        \E[Z_1]=\int_0^1{h^\lambda_*}'(t)dt=\int_0^1\{{\tilde{g}_*}'(t)+2(\gamma+\lambda)G^{-1}(t)\}dt=\E[Z_2],
    \end{align*}
    and
    \begin{align*}
        \int_\alpha^1F_{Z_1}^{-1}(t)dt=\ &\int_{\alpha}^1{h^\lambda_*}'(t)dt \\
        =\ &g(1)+2(\gamma+\lambda)\mu_G-h^\lambda_*(\alpha) \\
        \le\ &g(1)+2(\gamma+\lambda)\mu_G-\left\{\tilde{g}_*(\alpha)+2(\gamma+\lambda)\int_0^\alpha G^{-1}(t)dt\right\} \\
        =\ &\int_\alpha^1\{{\tilde{g}_*}'(t)+2(\gamma+\lambda)G^{-1}(t)\}dt \\
        =\ &\int_\alpha^1F_{Z_2}^{-1}(t)dt
    \end{align*}
    for all $\alpha\in[0,1]$. By Theorem 3.A.5 of \cite{Shaked2007Stochastic}, we have $Z_1 \leq_{\rm cx} Z_2$. 
	As such, 
    \begin{equation}\label{eq:finite}
    \E[Z_1^2]=\int_0^1\left({h^\lambda_*}'(t)\right)^2dt\le\E[Z_2^2]=\int_0^1\left({\tilde{g}_*}'(t)+2(\gamma+\lambda)G^{-1}(t)\right)^2dt<\infty,
    \end{equation}
    where the second inequality is due to Assumption \ref{ass:nonconstant}.

For any $\delta > 0$, define
\begin{equation*}
f^\lambda(t) := \frac{h^\lambda_*(\delta)}{\delta} t \1_{\{0 \leq t \leq \delta\}} + h^\lambda_*(t) \1_{\{\delta < t \leq 1\}},    ~~~ 
k^\lambda(t) := \frac{r^\lambda(\delta)}{\delta} t \1_{\{0 \leq t \leq \delta\}} +  r^\lambda(t) \1_{\{\delta < t \leq 1\}},
\end{equation*}
where $r^\lambda(t) = \Tilde{g}_*(t) + 2(\gamma+\lambda)\int_0^tG^{-1}(s)ds$. 
Moreover, $r^\lambda(t) \leq h^\lambda_*(t) \leq h^\lambda(t)$ and $r^\lambda(0) = h^\lambda_*(0) = h^\lambda(0) = 0$. 
It can be easily verified that $f^\lambda$ and $k^\lambda$ are continuous convex functions on $[0, 1]$ and $k^\lambda \leq f^\lambda$.
Moreover, $f^\lambda(0) = k^\lambda(0) = 0$ and $f^\lambda(1) = k^\lambda(1)$.
Note that for any $\alpha \in (0, 1)$,
\begin{equation*}
\int_0^\alpha \left(k^\lambda\right)'(u) du = k^\lambda(\alpha) \leq f^\lambda(\alpha) = \int_0^\alpha \left(f^\lambda\right)'(u) du.
\end{equation*}
Therefore, $\left(f^\lambda\right)'(U) \leq_{\rm cx} \left(k^\lambda\right)'(U)$.
Hence, we have
$$
    \int_0^1 \left[\left(f^\lambda\right)'(u)\right]^2 du \leq \int_0^1 \left[\left(k^\lambda\right)'(u)\right]^2 du.
$$
It can be rewritten as
\begin{align*}
	\int_\delta^1 \left[\left(h^\lambda_*\right)'(u)\right]^2 du &\leq \int_\delta^1 \left[\left(\Tilde{g}_*(u)\right)' + 2(\gamma+\lambda) G^{-1}(u)\right]^2 du + \frac{\left(r^\lambda(\delta)\right)^2}{\delta} - \frac{\left(h^\lambda_*(\delta)\right)^2}{\delta} \\
	&\leq \int_\delta^1 \left[|\left(\Tilde{g}_*(u)\right)'| + 2(\gamma+\lambda) |G^{-1}(u)|\right]^2 du + \frac{|r^\lambda(\delta) + h^\lambda_*(\delta)| |r^\lambda(\delta) - h^\lambda_*(\delta)|}{\delta}   \\
	&\leq \int_\delta^1 \left[|\left(\Tilde{g}_*(u)\right)'| + 2(\gamma+\lambda) |G^{-1}(u)|\right]^2 du + \frac{\left(|r^\lambda(\delta)| + |h^\lambda_*(\delta)|\right) |\Tilde{g}_*(\delta) - \Tilde{g}(\delta)|}{\delta},
\end{align*}
where the final inequality follows from $r^\lambda \leq h^\lambda_* \leq h^\lambda$.
For a given $\lambda_0 > 0$, 
\begin{align*}
\limsup_{\delta\to 1}\frac{|r^\lambda(\delta)|}{\delta} =&\limsup_{\delta\to 1} \frac{\left|\Tilde{g}_*(\delta) + 2(\gamma+\lambda)\int_0^\delta G^{-1}(s)ds \right|}{\delta}\\
\leq&\limsup_{\delta\to 1} \frac{\left|\Tilde{g}_*(\delta) \right|}{\delta} + 2 (\gamma+\lambda_0)\limsup_{\delta\to 1}\frac{\int_0^\delta \left|G^{-1}(s)\right| ds}{\delta} \\
=& |g(1)|+2(\gamma+\lambda_0)\int_0^1|G^{-1}(s)|ds<\infty.
\end{align*}
Thus, there exists a large $M > 0$ and $\delta_0 > 0$ such that when $\delta > \delta_0$
\begin{align*}
\sup_{0 \leq \lambda \leq \lambda_0} \left(\frac{\left(|r^\lambda(\delta)| + |h^\lambda(\delta)|\right)}{\delta}\right) \leq M.
\end{align*}
Consequently, for any $\eta > 0$, there exists a $\delta_0 > 0$ such that when $\delta > \delta_0$
\begin{eqnarray*}
& & \sup_{0 \leq \lambda \leq \lambda_0} \int_\delta^1 \left[\left({h^{\lambda}_*}\right)'(t)\right]^2dt \\
& &~\leq \sup_{0 \leq \lambda \leq \lambda_0} \left(\int_\delta^1 \left[|\left(\Tilde{g}_*(u)\right)'| + 2(\gamma+\lambda) |G^{-1}(u)|\right]^2 du + \frac{\left(|r^\lambda(\delta)| + |h^\lambda_*(\delta)|\right) |\Tilde{g}_*(\delta) - \Tilde{g}(\delta)|}{\delta}\right)  \\
& &~\leq \int_\delta^1 \left[|\left(\Tilde{g}_*(u)\right)'| + 2(\gamma+\lambda_0) |G^{-1}(u)|\right]^2 du + \sup_{0 \leq \lambda \leq \lambda_0} \left(\frac{\left(|r^\lambda(\delta)| + |h^\lambda(\delta)|\right) |\Tilde{g}_*(\delta) - \Tilde{g}(\delta)|}{\delta}\right) \\
& &~\leq \int_\delta^1 \left[|\left(\Tilde{g}_*(u)\right)'| + 2(\gamma+\lambda_0) |G^{-1}(u)|\right]^2 du + M |\Tilde{g}_*(\delta) - \Tilde{g}(\delta)|<\eta.
\end{eqnarray*}
Hence, $\left\{\left[\left({h^\lambda_*}\right)'(t)\right]^2, 0 \leq \lambda \leq \lambda_0\right\}$ is uniformly integrable for any $\lambda_0 > 0$. 

Similar to \eqref{eq:finite}, for each $\lambda_n$, we have $\int_0^1\left({h^{\lambda_n}_*}'(t)\right)^2dt<\infty$. When $\lambda_n\to \lambda$, where $\lambda<\infty$, it is straightforward that $\{\lambda_n\}$ is bounded by some $L\in(0,\infty)$, which results in
\begin{equation*}
\int_0^1\left({h^{\lambda_n}_*}'(t)\right)^2dt \le C:= \int_0^1\left({\tilde{g}_*}'(t)+2(\gamma+L)|G^{-1}(t)|\right)^2dt<\infty,
\end{equation*}
Therefore,
\begin{eqnarray*}
& & \int_0^1\left((h^{\lambda_n}_*)'(t)\right)^2dt-\int_0^1\left((h^\lambda_*)'(t)\right)^2dt \\
& &~= \int_0^1 \left((h^{\lambda_n}_*)'(t)-(h^{\lambda}_*)'(t)\right) \left((h^{\lambda_n}_*)'(t)+(h^{\lambda}_*)'(t)\right)dt \\ 
& &~\le \left(\int_0^1((h^{\lambda_n}_*)'(t)- (h^{\lambda}_*)'(t))^2dt\right)^{\frac{1}{2}} \left(\int_0^1((h^{\lambda_n}_*)'(t)+(h^{\lambda}_*)'(t))^2dt\right)^{\frac{1}{2}} \\
& &~\le \left(\int_0^1((h^{\lambda_n}_*)'(t)-(h^{\lambda}_*)'(t))^2dt\right)^{\frac{1}{2}} \left[\left(\int_0^1((h^{\lambda_n}_*)'(t))^2dt\right)^{\frac{1}{2}}+\left(\int_0^1((h^{\lambda}_*)'(t))^2dt\right)^{\frac{1}{2}}\right] \\
& &~\le
\left(\int_0^1((h^{\lambda_n}_*)'(t)-(h^{\lambda}_*)'(t))^2dt\right)^{\frac{1}{2}} \left[C+\left(\int_0^1((h^{\lambda}_*)'(t))^2dt\right)^{\frac{1}{2}}\right] \longrightarrow 0.
\end{eqnarray*}
The proof is complete.
\end{proof}

We can now characterize the worst-case distribution in \eqref{Prob:sup2}.
\begin{theorem}\label{thm:wo_quantile2}
    Under Assumption~\ref{ass:nonconstant}, the quantile function of the worst-case distribution $F_{wo}$ for problem \eqref{Prob:sup2} is
    \begin{equation*}
        {F_{wo}}^{-1}(t)=\mu_G+\frac{g(1)}{2(\gamma+\lambda^*)}+\frac{(h^{\lambda^*}_*)'(t)-h^{\lambda^*}(1)}{2\lambda^*},
    \end{equation*}
    where there exists an interval $[\eta,M]\subseteq (0,\infty)$ such that
    \begin{equation}\label{eq:lambda_minimizer}
        \lambda^*=\argmin_{\lambda\in[\eta,M]}\ \frac{g(1)^2}{4(\gamma+\lambda)}+\frac{\Lambda(\lambda)^2}{4\lambda}-\lambda(\sigma_G^2-\epsilon).
    \end{equation}
\end{theorem}
\begin{proof}
    By Lemma \ref{lem:binding2}, the Wasserstein constraint is binding at optimality. Using the binding constraint to eliminate $\int_0^1(F^{-1})^2dt$ from the objective in \eqref{Prob:sup2} gives
    $$
    \rho_g^F(X)+\gamma\left(\E[X^2]-\E[X]^2\right)=\rho_g^F(X)+\gamma\left(2\int_0^1F^{-1}(t)G^{-1}(t)dt+\epsilon-\int_0^1(G^{-1})^2dt-\E[X]^2\right).
    $$
    Thus, problem \eqref{Prob:sup2} becomes
    \begin{equation}\label{eq:prob_reduced}
        \begin{aligned}
            \sup_{F}&\ \int_0^1F^{-1}(t)d\tilde{g}(t)+2\gamma\int_0^1F^{-1}(t)G^{-1}(t)dt-\gamma\left(\int_0^1F^{-1}(t)dt\right)^2 \\
            \text{s.t.}&\ \int_0^1\left(F^{-1}(t)-G^{-1}(t)\right)^2dt\le \epsilon.
        \end{aligned}
    \end{equation}
    The objective in \eqref{eq:prob_reduced} is concave in $F^{-1}$, the constraint is convex, and $F^{-1}=G^{-1}$ is strictly feasible. Slater's condition gives equality with the dual value
    \begin{equation}\label{eq:prob_reduced2}
    \begin{aligned}
        \inf_{\lambda\ge 0}\sup_F&\ \int_0^1F^{-1}(t)d\tilde{g}(t)+2\gamma\int_0^1F^{-1}(t)G^{-1}(t)dt-\gamma\left(\int_0^1F^{-1}(t)dt\right)^2 \\
        &\ ~~~~ -\lambda\left(\int_0^1\left(F^{-1}(t)-G^{-1}(t)\right)^2dt-\epsilon\right).
    \end{aligned}
    \end{equation}
    A quantile function with mean $\mu_F$ and variance $\sigma_F^2$ can be written as $\mu_F+\sigma_F\tilde F^{-1}$, where $\tilde F^{-1}$ has mean zero and variance one. This separates location and scale from the quantile shape, which will be optimized using the convex envelope. Thus \eqref{eq:prob_reduced2} can be rewritten as
    \begin{eqnarray}
        & & \inf_{\lambda\ge0}\sup_{\mu_F,\sigma_F}\sup_{\tilde{F}}\ \mu_Fg(1)+\sigma_F\int_0^1\tilde{F}^{-1}(t)d\tilde{g}(t)+2\gamma\left(\mu_F\mu_G+\sigma_F\int_0^1\tilde{F}^{-1}(t)G^{-1}(t)dt\right) \nonumber \\
        & &\quad\quad\quad\quad\quad\ ~~~~~ -\gamma\mu_F^2-\lambda\left(\int_0^1\left(\mu_F+\sigma_F\tilde{F}^{-1}(t)-G^{-1}(t)\right)^2dt-\epsilon\right) \nonumber \\
        & &\ ~~ = \inf_{\lambda\ge 0}\sup_{\mu_F,\sigma_F}\sup_{\tilde{F}}\ \mu_Fg(1)+\sigma_F\int_0^1\tilde{F}^{-1}(t)d\tilde{g}(t)+2\gamma\left(\mu_F\mu_G+\sigma_F\int_0^1\tilde{F}^{-1}(t)G^{-1}(t)dt\right) \nonumber \\
        & &\quad\quad\quad\quad\quad\ ~~~~~~~~~~~~ -\gamma\mu_F^2-\lambda\left((\mu_F-\mu_G)^2+\sigma_F^2+\sigma_G^2-2\sigma_F\int_0^1\tilde{F}^{-1}(t)G^{-1}(t)dt-\epsilon\right) \nonumber \\
        & &\ ~~ = \inf_{\lambda\ge 0}\Bigg\{\sup_{\mu_F,\sigma_F}\Big\{ \mu_Fg(1)+2\gamma\mu_F\mu_G-\gamma\mu_F^2-\lambda\left((\mu_F-\mu_G)^2+\sigma_F^2+\sigma_G^2-\epsilon\right) \nonumber \\
        & &\quad\quad\quad\quad\quad\ +\sigma_F\ \sup_{\tilde{F}}\big\{\int_0^1\tilde{F}^{-1}(t)d\tilde{g}(t)+2(\gamma+\lambda)\int_0^1\tilde{F}^{-1}(t)G^{-1}(t)dt\big\}\Big\}\Bigg\}. \label{eq:prob_reduced3}
    \end{eqnarray}
    For the inner problem, we have
    \begin{align*}
        &\int_0^1\tilde{F}^{-1}(t)d\tilde{g}(t)+\int_0^1\tilde{F}^{-1}(t)d\left[2(\gamma+\lambda) \int_0^tG^{-1}(s)ds \right] \\
        &= \int_0^1\tilde{F}^{-1}(t)dh^\lambda(t) 
        \le \int_0^1\tilde{F}^{-1}(t)(h^\lambda_*)'(t)dt = \int_0^1\tilde{F}^{-1}(t)\left((h^\lambda_*)'(t)-h^\lambda(1)\right)dt \\
        &\le \left(\int_0^1((h^\lambda_*)'(t)-h^\lambda(1))^2dt\right)^{\frac{1}{2}}= \Lambda(\lambda),
    \end{align*}
where the first inequality follows from Theorem \ref{thm:modified} and the second from the Cauchy--Schwarz inequality.
Equality in the second inequality holds when
$$\tilde{F}^{-1}(t)=\frac{(h^\lambda_*)'(t)-h^\lambda(1)}{\sqrt{\int_0^1((h^\lambda_*)'(t)-h^\lambda(1))^2dt}},$$ 
which, by Lemma~\ref{lem:envelope}, is constant on each interval where $h^\lambda_*(t)<h^\lambda(t-)\wedge h^\lambda(t+)$.
The equality condition in Theorem~\ref{thm:modified} is therefore satisfied. Moreover, note that
\begin{equation*}
\int_0^1\big((h^\lambda_*)'(t)-h_\lambda(1)\big) dt
=h^\lambda_*(1)-h^\lambda_*(0)-h^\lambda(1)\
=0,
\end{equation*}
where we use $h^\lambda_*(0)=0$ and $h^\lambda_*(1)=h^\lambda(1)$. Hence,
$\int_0^1\tilde F^{-1}(t)dt=0$.
By the definition of the normalization,
$\int_0^1\big(\tilde F^{-1}(t)\big)^2 dt=1.$
Thus $\tilde F^{-1}$ has mean zero and variance one, and therefore is admissible for the inner optimization problem. Hence
\begin{equation}\label{eq-sup}
\Lambda(\lambda)=\sup_{\tilde{F}}\ \left\{\int_0^1\tilde{F}^{-1}(t)d\tilde{g}(t)+2(\gamma+\lambda)\int_0^1\tilde{F}^{-1}(t)G^{-1}(t)dt\right\}.
\end{equation}
Substituting this optimal value into \eqref{eq:prob_reduced3} yields
\begin{eqnarray}
    & & \inf_{\lambda\ge 0}\left\{\sup_{\mu_F,\sigma_F}\mu_Fg(1)+2\gamma\mu_F\mu_G-\gamma\mu_F^2-\lambda\left((\mu_F-\mu_G)^2+\sigma_F^2+\sigma_G^2-\epsilon\right)+\sigma_F\Lambda(\lambda)\right\} \label{eq:prob_reduced4} \\
    & & \ ~~ = \inf\Bigg\{\inf_{\lambda> 0}\left\{\sup_{\mu_F,\sigma_F}\mu_Fg(1)+2\gamma\mu_F\mu_G-\gamma\mu_F^2-\lambda\left((\mu_F-\mu_G)^2+\sigma_F^2+\sigma_G^2-\epsilon\right)+\sigma_F\Lambda(\lambda)\right\}, \nonumber\\
    & &~~~~~~~\quad\quad\quad \sup_{\mu_F,\sigma_F}\ \mu_Fg(1)+2\gamma\mu_F\mu_G-\gamma\mu_F^2+\sigma_F\Lambda(0)\Bigg\}. \nonumber
\end{eqnarray}
Note that $\Lambda(0)>0$ under Assumption \ref{ass:nonconstant}, leading to 
$$
\sup_{\mu_F,\sigma_F}\ \left\{ \mu_Fg(1)+2\gamma\mu_F\mu_G-\gamma\mu_F^2+\sigma_F\Lambda(0) \right\} = \infty.
$$
Let $\mathcal L(\mu_F,\sigma_F,\lambda)$ denote the objective in \eqref{eq:prob_reduced4}. Direct differentiation gives
\begin{align*}
    &\frac{\partial^2 \mathcal{L}(\mu_F,\sigma_F,\lambda)}{\partial\mu_F^2}=-2(\gamma+\lambda),\quad \frac{\partial^2\mathcal{L}(\mu_F,\sigma_F,\lambda)}{\partial\mu_F\sigma_F}=0,\quad \frac{\partial^2\mathcal{L}(\mu_F,\sigma_F,\lambda)}{\partial\sigma_F^2}=-2\lambda.
\end{align*}
Hence, for any $\lambda>0$, the Hessian matrix of $\mathcal{L}(\mu_F,\sigma_F,\lambda)$ is negative definite. Thus, for $\lambda>0$, $\mathcal L(\mu_F,\sigma_F,\lambda)$ has the unique maximizer
$$(\mu_F^*,\sigma_F^*) = \left(\mu_G+\frac{g(1)}{2(\gamma+\lambda)}, \frac{\Lambda(\lambda)}{2\lambda} \right).$$ 
With these results, the problem \eqref{eq:prob_reduced4} further reduces to
\begin{equation}\label{eq-inf}
    \inf_{\lambda>0}\ \left\{\frac{g(1)^2}{4(\gamma+\lambda)}+\frac{\Lambda(\lambda)^2}{4\lambda}-\lambda(\sigma_G^2-\epsilon)\right\}+\mu_G g(1)+\gamma\mu_G^2.
\end{equation}
Taking $F=G$ in the inner problem of \eqref{eq:prob_reduced2} gives
\[
\frac{g(1)^2}{4(\gamma+\lambda)}
+\frac{\Lambda(\lambda)^2}{4\lambda}
-\lambda(\sigma_G^2-\epsilon)
\geq\lambda\epsilon+\rho_g^G(X)-\mu_G g(1)+2\gamma\sigma_G^2.
\]
The right-hand side tends to $+\infty$ as $\lambda\to\infty$, since $\epsilon>0$. At zero, continuity from Lemma~\ref{lem:continuity_Lambda} and $\Lambda(0)>0$ imply that the objective tends to $+\infty$ as $\lambda\downarrow0$. Thus the objective is continuous and coercive on $(0,\infty)$, so its minimum is attained on a compact interval $[\eta,M]\subset(0,\infty)$. 

It remains to verify that the constructed $F_{wo}$ is feasible for problem \eqref{Prob:sup2}. Since $h^{\lambda^*}_*$ is convex, $(h^{\lambda^*}_*)'$ is increasing. Hence $F_{wo}^{-1}$ is an increasing $L^2([0,1])$ function and therefore defines a probability distribution with finite second moment. 
Moreover, by the equality conditions established in \eqref{eq-sup}, $F_{wo}$ attains the inner supremum corresponding to $\lambda^*$. Since $\lambda^*$ minimizes the dual objective in \eqref{eq-inf}, strong duality implies that $F_{wo}$ is a primal optimizer. Therefore, by Lemma \ref{lem:binding2}, the Wasserstein constraint is binding at $F_{wo}$, and
\begin{equation*}
\int_0^1\left(F_{wo}^{-1}(t)-G^{-1}(t)\right)^2dt=\epsilon.
\end{equation*}
The proof is now finished.
\end{proof}

When $g$ is concave, $\tilde g$ is convex and hence $h^\lambda=h^\lambda_*$ on $[0,1]$. The formula in Theorem \ref{thm:wo_quantile2} then simplifies as follows. 

\begin{corollary}\label{cor:concave-distortion}
    If $g$ is concave, then the quantile function of the worst-case distribution $F_{wo}$ for problem \eqref{Prob:sup2} is
    \begin{equation}
        F_{wo}^{-1}(t)=\mu_G+\frac{g(1)}{2(\gamma+\lambda^*)}+\frac{(h^{\lambda^*})'(t)-h^{\lambda^*}(1)}{2\lambda^*},
    \end{equation}
    where $\lambda^*=\inf\left\{\lambda>0: \mathcal{K}'(\lambda)\ge 0\right\}$ and
    \begin{equation}
        \mathcal{K}(\lambda)=\frac{g(1)^2}{4(\gamma+\lambda)}+\frac{A+4\gamma B+4\gamma^2\sigma_G^2}{4\lambda}+\lambda\epsilon,
    \end{equation}
    where $A=\int_0^1(\tilde{g}'(t)-g(1))^2dt$ and $B=\int_0^1(\tilde{g}'(t)-g(1))(G^{-1}(t)-\mu_G)dt$.
\end{corollary}

For the best-case problem \eqref{Prob:inf2}, the binding boundary $W_2(F,G)=\sqrt\epsilon$ gives
\begin{align*}
    &\inf_{\mu_F\in\R}\inf_{F\in\mathcal{M}^2(G;\epsilon)}\ \rho_g^F(X)+\gamma(\E_F[X^2]-\mu_F^2) \\
    =& \inf_{\mu_F\in\R}\inf_{F\in\mathcal{M}^2(G;\epsilon)}\ \left\{\rho_g^F(X)+\gamma\E_F[X^2]\right\}-\mu_F^2 \\
    =& \inf_{\mu_F\in\R}\inf_{F\in\mathcal{M}^2(G;\epsilon)}\left\{\int_0^1F^{-1}(t)d\tilde{g}(t)+2\gamma\int_0^1F^{-1}(t)G^{-1}(t)dt\right\}+\epsilon-\int_0^1(G^{-1}(t))^2dt-\mu_F^2 \\
    =& \inf_{\mu_F\in\R}\left\{-\sup_{F\in\mathcal{M}^2(G;\epsilon)}-\left(\int_0^1F^{-1}(t)d[\tilde{g}(t)+2\gamma\int_0^tG^{-1}(s)ds]\right)\right\}+\epsilon-\int_0^1(G^{-1}(t))^2dt-\mu_F^2 \\
    =&\inf_{\mu_F\in\R}\left\{-\sup_{F\in\mathcal{M}^2(G;\epsilon)}\int_0^1F^{-1}(t)d[-h(t)]\right\}+\epsilon-\int_0^1(G^{-1}(t))^2dt-\mu_F^2.
\end{align*}
Note that the inner problem is to maximize the integral $\int_0^1F^{-1}(t)d[-h(t)]$ with $-h\in\mathcal{G}$. As such, the methodology used in Theorem \ref{thm:wo_quantile2} is still applicable if Assumption \ref{ass:nonconstant} applies to the function $(-h)_*'$. We therefore do not pursue the solution to the problem \eqref{Prob:inf2}, as it closely resembles that of the problem \eqref{Prob:sup2}.

\subsection{Application to RVaR}
For $0\leq\alpha<\beta\leq1$, the RVaR distortion function is
\begin{equation}
    g(t):=
    \min\left\{
        \frac{(t-1+\beta)_+}{\beta-\alpha},1
    \right\},
    \qquad t\in[0,1].
\end{equation}
RVaR averages the loss quantiles between probability levels $\alpha$ and $\beta$. Here $\alpha$ is a probability level, distinct from the ambiguity-attitude parameter in Subsection~\ref{sec:optimal-sharing}. RVaR includes several familiar tail risk measures; for $\beta=1$, it is TVaR at level $\alpha$. Its piecewise-linear distortion allows us to examine the effect of such nonsmoothness on the extremal law.

For this distortion, the function $h^\lambda$ in Theorem~\ref{thm:wo_quantile2} is
\begin{equation}
h^\lambda(t)=\begin{cases}
\displaystyle
2(\gamma+\lambda)\int_0^t G^{-1}(s)\,ds,
& 0\leq t\leq\alpha,\\[2ex]
\displaystyle
\frac{t-\alpha}{\beta-\alpha}+2(\gamma+\lambda)\int_0^t G^{-1}(s)\,ds,
& \alpha<t\leq\beta,\\[2ex]
\displaystyle
1+2(\gamma+\lambda)\int_0^t G^{-1}(s)\,ds,
& \beta<t\leq1.
\end{cases}
\label{eq:h_lambda_RVaR}
\end{equation}
The worst-case quantile is determined by the convex envelope of $h^\lambda$. If $\beta=1$, then $h^\lambda$ is convex on $[0,1]$, so $h^\lambda_*=h^\lambda$ and Corollary~\ref{cor:concave-distortion} applies. We therefore consider $\beta<1$.

The function $h^\lambda$ is convex on each of the intervals $[0,\beta]$ and $[\beta,1]$. However, its derivative has a downward jump at $\beta$:
\[
    (h^\lambda)'(\beta-)=\frac{1}{\beta-\alpha}+2(\gamma+\lambda)G^{-1}(\beta)
    >2(\gamma+\lambda)G^{-1}(\beta)=(h^\lambda)'(\beta+).
\]
Consequently, $h^\lambda$ is generally not convex on $[0,1]$. Its convex envelope is obtained by replacing $h^\lambda$ with a supporting line segment on an interval containing $\beta$. More precisely, there
exist $a\in[\alpha,\beta]$ and $b\in[\beta,1]$, with $a<b$, such that
\begin{equation}
h^{\lambda}_*(t)=
\begin{cases}
h^\lambda(t), & 0\leq t\leq a,\\
h^\lambda(a)+K(t-a), & a<t<b,\\
h^\lambda(t), & b\leq t\leq1,
\end{cases}
\label{eq:convex_envelope_RVaR}
\end{equation}
where
\begin{equation}\label{eq1:slope_K}
    K=\frac{h^\lambda(b)-h^\lambda(a)}{b-a}.
\end{equation}
In this construction, the convexification of $h^\lambda$ is confined to the interval $[a,b]$ containing $\beta$.

Suppose that the contact points are interior, with $a\in(\alpha,\beta)$ and $b\in(\beta,1)$. The supporting line is tangent to $h^\lambda$ at both endpoints, so
\begin{equation}
    K=(h^\lambda)'(a+)=(h^\lambda)'(b-),
\end{equation}
or equivalently,
\begin{equation}
    K=\frac{1}{\beta-\alpha}+2(\gamma+\lambda)G^{-1}(a)=2(\gamma+\lambda)G^{-1}(b).
    \label{eq:tangency_RVaR}
\end{equation}
Together with \eqref{eq1:slope_K}, we obtain the following equations for $a$ and $b$:
\begin{equation}
    G^{-1}(b)-G^{-1}(a)=\frac{1}{2(\gamma+\lambda)(\beta-\alpha)},
    \label{eq:RVaR_ab_1}
\end{equation}
and
\begin{equation}
    2(\gamma+\lambda)\int_a^b\bigl(G^{-1}(s)-G^{-1}(a)\bigr)\,ds=\frac{b-\beta}{\beta-\alpha}.
    \label{eq:RVaR_ab_2}
\end{equation}
In particular, the convexification depends on the reference distribution through
its quantile function and on the parameters $(\gamma,\lambda,\alpha,\beta)$.

An interior solution need not exist. In that case, the supporting line may touch $h^\lambda$ at one of the boundary points $\alpha$ or $1$. There are three possible boundary configurations.

\begin{itemize}
    \item If $a=\alpha$ and $b\in(\beta,1)$, then $b$ solves
    \begin{equation}
        2(\gamma+\lambda)G^{-1}(b)=\frac{h^\lambda(b)-h^\lambda(\alpha)}{b-\alpha},
        \end{equation}
    with the supporting slope satisfying
    \begin{equation}
        G^{-1}(\alpha)\leq G^{-1}(b)\leq G^{-1}(\alpha)+\frac{1}{2(\gamma+\lambda)(\beta-\alpha)}.
    \end{equation}

    \item If $a\in(\alpha,\beta)$ and $b=1$, then $a$ solves
    \begin{equation}
        \frac{1}{\beta-\alpha}+2(\gamma+\lambda)G^{-1}(a)=\frac{h^\lambda(1)-h^\lambda(a)}{1-a},
    \end{equation}
    with
    \begin{equation}
        G^{-1}(a)\geq G^{-1}(1)-\frac{1}{2(\gamma+\lambda)(\beta-\alpha)}.
    \end{equation}

    \item If $a=\alpha$ and $b=1$, then
    \begin{equation}
        K=\frac{h^\lambda(1)-h^\lambda(\alpha)}{1-\alpha},
    \end{equation}
    and the corresponding supporting-slope conditions are
    \begin{equation}
        2(\gamma+\lambda)G^{-1}(1)
        \leq K \leq \frac{1}{\beta-\alpha}
        +2(\gamma+\lambda)G^{-1}(\alpha).
    \end{equation}
\end{itemize}
When $b=1$, the supporting line must lie below $h^\lambda$ to the left of the endpoint, which requires $K\geq(h^\lambda)'(1-)=2(\gamma+\lambda)G^{-1}(1)$. Hence, if $G^{-1}(1)=+\infty$, neither boundary configuration is possible with a finite slope. Moreover, the case $a=b=\beta$ cannot occur, since the left derivative of $h^\lambda$ at $\beta$ is strictly larger than its right derivative, and therefore no supporting line can be tangent to $h^\lambda$ at $\beta$.

The preceding characterization has a useful interpretation. The non-convexity generated by the RVaR distortion is caused entirely by the
downward jump in the slope at $\beta$. Convexification replaces $h^\lambda$ on an interval containing $\beta$ with a line segment, so the derivative of the convex envelope is constant on that interval. Substituting the derivative of this convex envelope into Theorem~\ref{thm:wo_quantile2} gives the worst-case quantile. At the optimal multiplier, the linear segment produces a flat portion of this quantile and hence an atom in the extremal distribution. For RVaR, the effect of the nonsmooth distortion can thus be computed from the contact points and the supporting slope.

\paragraph{Numerical illustration.}
We take $\alpha=0.1$ and $\beta=0.9$, with an exponential reference distribution of mean one. We vary the ambiguity-set parameter $\epsilon$ and the variance-aversion coefficient $\gamma$ separately, keeping the other parameters fixed. The two panels examine the effects of the ambiguity level, governed by $\epsilon$, and the relative weight assigned to variance versus RVaR, governed by $\gamma$, on the quantile functions of the resulting extremal distributions. The dashed curve in each panel is the reference quantile.

\begin{figure}[!htbp]
\centering
\minipage{0.5\textwidth}
 \centering
  \includegraphics[width=\linewidth]{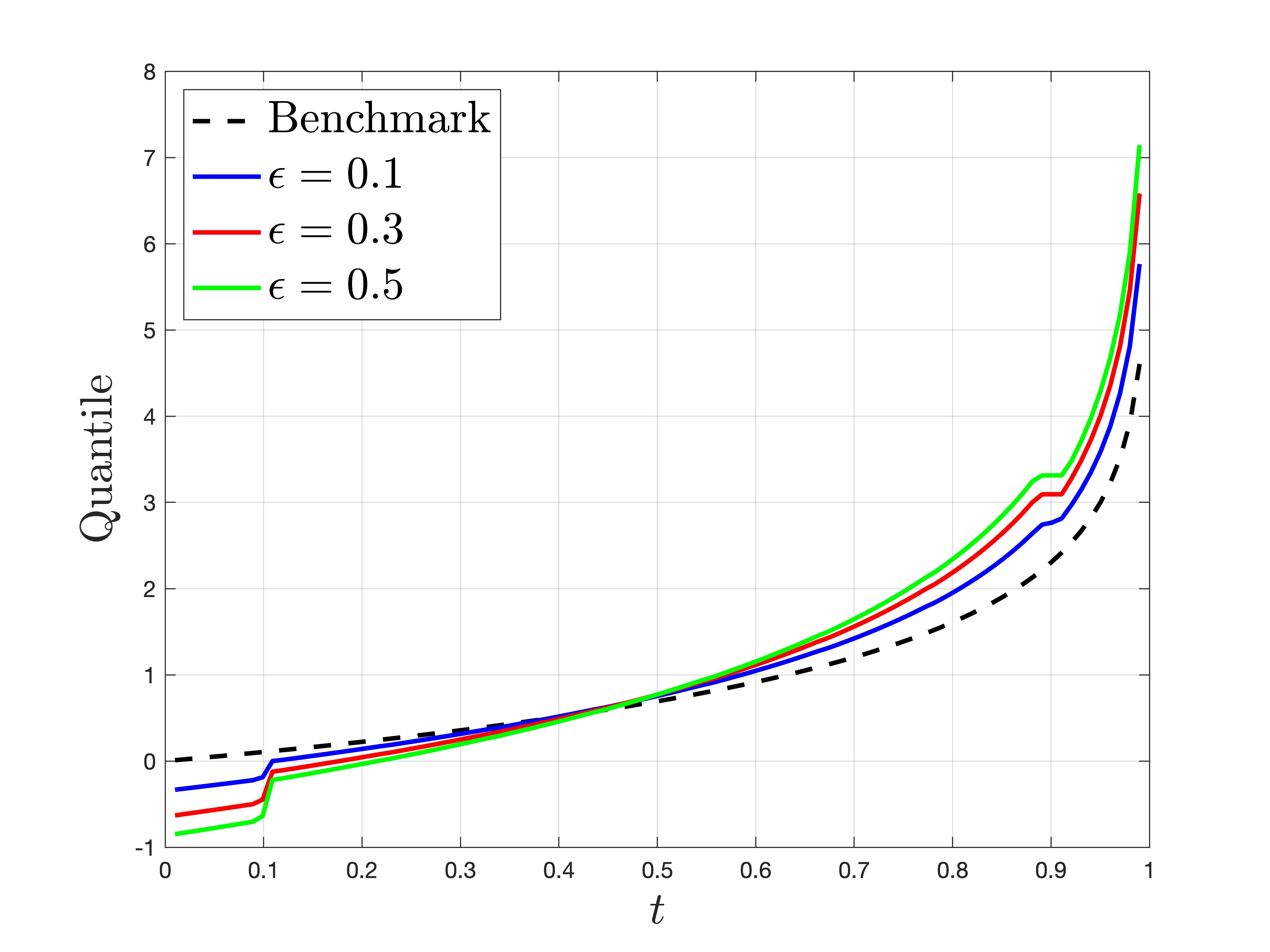}
\endminipage
\minipage{0.5\textwidth}
  \centering
  \includegraphics[width=\linewidth]{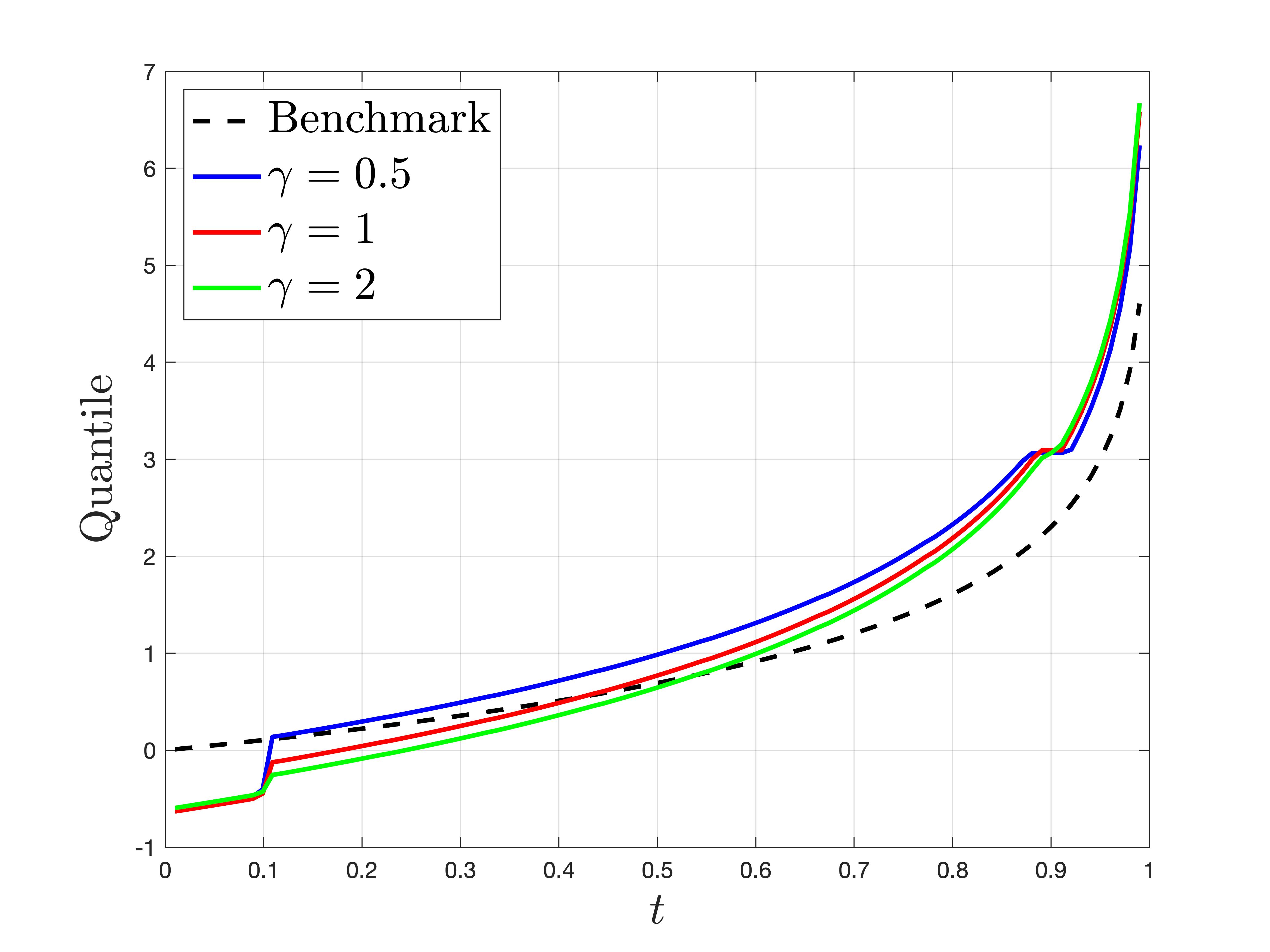}
\endminipage
\caption{{Worst-case quantile functions for different values of $\epsilon$ with $\gamma=1$ (left) and for different variance-aversion coefficients with $\epsilon=0.3$ (right).}}
\label{fig:RVaR_sens}
\end{figure}

The left panel shows the effect of the ambiguity radius. As $\epsilon$ increases, the worst-case distribution can move farther from the reference distribution. This leads to lower quantiles over part of the lower and central range and substantially higher quantiles in the upper tail. The quantile curves therefore cross rather than being ordered pointwise. This is not inconsistent with the fact that the worst-case objective is nondecreasing in $\epsilon$: enlarging the ambiguity set increases the maximum value, but it need not increase every quantile of the maximizing distribution. The negative lower quantiles are also admissible because the Wasserstein ball does not impose the nonnegative support of the exponential reference distribution; under our sign convention, negative losses correspond to gains.

The right panel isolates the effect of variance aversion by fixing $\epsilon=0.3$. Increasing $\gamma$ gives greater weight to variance relative to the RVaR component. For the parameter values considered here, this lowers the quantiles over much of the distribution while increasing the far upper tail. Thus, greater variance aversion does not simply shift the worst-case distribution upward. Instead, it changes its shape by allowing larger extreme losses while reducing losses over a substantial part of the probability range. This reflects the fact that the variance term depends on the entire distribution.

The figure also illustrates the effect of the nonsmooth RVaR distortion. The upward jump in the worst-case quantile near $\alpha$ is a consequence of the change in the slope of the RVaR distortion at $\alpha$. The flat portion near $\beta$, on the other hand, results from the linear segment introduced by the convexification and corresponds to an atom in the extremal distribution.

\section{\texorpdfstring{Conclusion}{Conclusion}}\label{sec:conclusion}
We characterize the worst- and best-case mean-variance values over a 2-Wasserstein ball and show that the extremal laws are location-scale transformations of the reference distribution. The extremal values depend only on the reference mean and variance, while the corresponding laws preserve the distributional shape. These characterizations also yield comparative statics with respect to variance aversion. We then apply the results to risk sharing. Under heterogeneous beliefs, proportional sharing reduces to a finite-dimensional dual problem involving agents' reference moments, ambiguity radii, and variance-aversion coefficients. Under homogeneous beliefs, the classical proportional allocation remains optimal for every candidate distribution and is therefore independent of both the ambiguity radius and the planner's ambiguity attitude, although the associated extremal costs vary.

We further extend the analysis to a coupled distortion-variance functional. In this setting, the extremal law generally depends on features beyond the first two moments and need not preserve the shape of the reference distribution. Under suitable conditions, we characterize the worst-case quantile through a convex-envelope construction. The RVaR example illustrates how distortion, Wasserstein ambiguity, and variance aversion jointly reshape the extremal distribution.
 
Several questions remain open. In particular, extending the location–scale characterization to general $p$-Wasserstein distances and obtaining a complete characterization of the best-case distortion–variance problem are natural directions for future work. Second, the present analysis is based on the first two moments of the loss. Incorporating higher-order distributional features into the decision criterion could provide a richer description of the effects of distributional uncertainty. Another important question is how to calibrate the Wasserstein radius from finite data and quantify the resulting statistical uncertainty in the evaluated costs.

\section*{Acknowledgements}
The authors thank Taizhong Hu and Tiantian Mao for valuable comments and suggestions. Wenjun Jiang acknowledges financial support from the Natural Sciences and Engineering Research Council of Canada (grant RGPIN-2020-04204), Alberta Innovates, and the Canadian Institute of Actuaries.
Yiying Zhang acknowledges support from the National Natural Science Foundation of China (grant 12571513) and the Shenzhen Science and Technology Program (grants JCYJ20230807093312026 and JCYJ20250604144210013).
Zhenfeng Zou acknowledges support from the National Natural Science Foundation of China (grant 12401625) and the Fundamental Research Funds for the Central Universities (grant WK2040000108).

\setlength{\bibsep}{0.0pt}
\bibliography{Robust_reference}

\appendix
\setcounter{equation}{0}
\renewcommand\theequation{A.\arabic{equation}}
\renewcommand\thedefinition{A.\arabic{definition}}
\renewcommand\thelemma{A.\arabic{lemma}}

\section{\texorpdfstring{{Auxiliary results}}{Auxiliary results}}\label{sec:A}

This section presents lemmas and theorems that support the results in the main text. Let ${\cal L}^0$ denote the space of all real-valued random variables on the probability space $(\Omega, {\cal F}, \Prob)$.

\begin{lemma}{\rm \citep[][Lemma 1]{wang2020distortion}}\label{lem:quantile-formulation}
    For $g\in\mathcal{G}$ and $X\in {\cal L}^0$ such that $\rho_g^F(X)$ is well defined, the following hold:
    \begin{itemize}
        \item[(\rmnum{1}).] if $g$ is right continuous, then $\rho_g^F(X)=\int_0^1F^{-1+}(1-t)dg(t)$;
        \item[(\rmnum{2}).] if $g$ is left continuous, then $\rho_g^F(X)=\int_0^1F^{-1}(1-t)dg(t)$;
        \item[(\rmnum{3}).] if $F^{-1}$ is continuous on $(0,1)$, then $\rho_g^F(X)=\int_0^1F^{-1}(1-t)dg(t)=\int_0^1F^{-1+}(1-t)dg(t)$.
    \end{itemize}
\end{lemma}

\begin{lemma}{\rm \citep[][Lemma B.1]{he2017rank}}\label{lem:envelope}
    Let $f:[0,1]\to\R$ be continuous. Its convex envelope $f_*$ satisfies the following properties:
    \begin{itemize}
        \item[(\rmnum{1}).] $f_*$ is continuous on $[0,1]$;
        \item[(\rmnum{2}).] $f_*$ is affine on $\{x\in[0,1]:\ f_*(x)<f(x)\}$;
        \item[(\rmnum{3}).] $f_*(0)=f(0)$ and $f_*(1)=f(1)$.
    \end{itemize}
\end{lemma}

\begin{theorem}{\rm \citep[][Theorem 1]{moriguti1953modification}}\label{thm:modified}
    Let $\Phi$ be of bounded variation on $[a,b]$ and continuous at both endpoints. Then
    \begin{equation}\label{ineq:modified}
        \int_a^bx(t)d\Phi(t)\le \int_a^bx(t)\Tilde{\varphi}(t)dt
    \end{equation}
    for every increasing function $x$ for which both integrals are finite, where $\Tilde{\varphi}$ is the right derivative of the convex envelope $\Bar\Phi$ of $\Phi$. Equality in \eqref{ineq:modified} holds if and only if $x$ is constant on every interval on which $\Bar\Phi(t)<\Phi(t-)\wedge\Phi(t+)$ and, at every discontinuity $t_n$ of $\Phi$,
    \begin{equation*}
        x(t_n)=\begin{cases}
            x(t_n+),& \text{if }\Phi(t_n-)<\Phi(t_n+), \\
            x(t_n-),& \text{if }\Phi(t_n-)>\Phi(t_n+).
        \end{cases}
    \end{equation*}
\end{theorem}

\begin{theorem}{\rm \citep[][Theorem 3.4]{sion1958general}}\label{thm:sion}
    Let $M$ be a compact convex set and $N$ a convex set. Suppose that, for each fixed $x\in M$, $f(x,\cdot)$ is upper semicontinuous and quasiconcave on $N$, and that, for each fixed $y\in N$, $f(\cdot,y)$ is lower semicontinuous and quasiconvex on $M$. Then the following minimax equality holds:
    \begin{equation*}
        \inf_{x \in M} \sup_{y \in N}\ f(x, y) = \sup_{y \in N} \inf_{x \in M}\ f(x, y).
    \end{equation*}
\end{theorem}

\section{Extension to higher dimensions}\label{sec:high_dim}
\setcounter{equation}{0}
\renewcommand\theequation{B.\arabic{equation}}

This appendix extends the one-dimensional results established in Theorems \ref{thm:wo_quantile} and \ref{thm:be_quantile} to high-dimensional settings. 
This extension is of practical relevance, as many decision-making problems, such as portfolio selection, involve multiple risk factors rather than a single loss variable.
Specifically, we consider a DM who is concerned with a linear combination $\bm w^\top \bm X$ of an $n$-dimensional random vector $\bm X =(X_1,\dots,X_n)^\top \in (\mathcal{L}^2)^n$, where $\bm w \in \R^n$ is a vector of coefficients.
In portfolio selection, for instance, $\bm w^\top$ is often restricted $\Delta_n$, though no such constraints are imposed in this subsection unless explicitly stated.

To proceed, we introduce the necessary notation for multivariate distributions.
Let $\bm X \in (\mathcal{L}^2)^n$ have joint distribution $F_{\bm X}$, and let $\bm Y \in (\mathcal{L}^2)^n$ be a benchmark random vector with joint distribution $F_{\bm Y}$. 
The $2$-Wasserstein distance between two multivariate distributions $F_{\bm X}$ and $F_{\bm Y}$ is defined as
\begin{equation*}
W_2^n(F_{\bm X}, F_{\bm Y}) = \inf_{\bm X \sim F_{\bm X}, \bm Y \sim F_{\bm Y}} \left(\E\left[\left\|\bm X - \bm Y \right\|_2^2\right]\right)^{1/2},
\end{equation*}
where $\|\cdot\|_2$ denotes the Euclidean norm on $\mathbb R^n$. 
This definition reduces to the univariate case \eqref{Wass} when $n=1$.

To simplify the analysis, we assume that the mean vector and covariance matrix of the reference random vector $\bm Y$ are known, denoted respectively by 
$$\E(\bm Y) = \bm \mu = (\mu_1, \ldots, \mu_n)^\top, ~~~ {\rm Cov}(\bm Y) = \Sigma,$$ 
where $\Sigma$ is assumed to be positive definite. 
This positive definiteness is a standard non-degeneracy condition in multivariate analysis, ensuring that no linear combination of the components is degenerate.

The multivariate Wasserstein ball of radius $\sqrt{\epsilon}$ centered at the reference distribution $F_{\bm Y}$ is then defined as
$$
{\cal M}^2(F_{\bm Y}; \epsilon) := \left\{F_{\bm X} \in (\mathcal{F}^2)^n: W_2^n(F_{\bm X}, F_{\bm Y}) \leq \sqrt{\epsilon}\right\}.
$$
The key technical tool that enables this extension is the projection property of the $2$-Wasserstein distance: for any fixed $\bm w \in \R^n$, the CDF of the linear combination $\bm w^\top \bm X$ lies in a one-dimensional Wasserstein ball centered at $\bm w^\top \bm Y$, with radius scaled by $\sqrt{\epsilon} \|\bm w\|_2$.
Formally,
\begin{equation}\label{eq-2026-314}
{\cal M}^2(F_{\bm Y}; \epsilon) = \left\{F \in (\mathcal{F}^2)^n: W_2^n(F, F_{\bm w^\top \bm Y}) \leq \sqrt{\epsilon} \left\|\bm w \right\|_2\right\} = {\cal M}^2\left(F_{\bm w^\top \bm Y}; \epsilon \left\|\bm w\right\|_2^2\right),
\end{equation}
This property, established in Theorem 5 of \cite{mao2026model}, allows us to reduce any high-dimensional robust optimization problem involving a linear combination of risk factors to the one-dimensional framework developed in Theorems \ref{thm:wo_quantile} and \ref{thm:be_quantile}.

With this projection property \eqref{eq-2026-314}, let $F$ be the CDF of $\widetilde{X}:=\bm w^\top \bm X$, the following proposition holds.
\begin{proposition}
For any fixed $\bm w \in \R^n$, the worst-case and best-case values of the mean-variance objective under the multivariate Wasserstein ambiguity set $\mathcal B(F_{\bm Y};\epsilon)$ are given respectively by
\begin{equation*}
\sup_{F_{\bm X} \in {\cal M}^2(F_{\bm Y}; \epsilon)} \left\{\E_{F_{\bm X}}[\bm w^\top \bm X] + \gamma \V_{F_{\bm X}} [\bm w^\top \bm X] \right\} = \sup_{F\in {\cal M}^2(F_{\bm w^\top \bm Y};\epsilon \left\|\bm w \right\|_2^2)} \left\{\E_F[\widetilde{X}]+\gamma\V_F[\widetilde{X}]]\right\},
\end{equation*}
and
\begin{equation*}
\inf_{F_{\bm X} \in {\cal M}^2(F_{\bm Y}; \epsilon)} \left\{\E_{F_{\bm X}}[\bm w^\top \bm X] + \gamma \V_{F_{\bm X}} [\bm w^\top \bm X] \right\} = \inf_{F\in {\cal M}^2(F_{\bm w^\top \bm Y};\epsilon \left\|\bm w \right\|_2^2)} \left\{\E_F[\widetilde{X}]+\gamma\V_F[\widetilde{X}]]\right\}.
\end{equation*}
\end{proposition}

\end{document}